\documentclass[a4paper,11pt]{article}

\title{Finding Competitive Equilibria in Online Marketplaces with Platform Fees, Taxes, and Subsidies}
\author{Taihei Oki\thanks{Institute for Chemical Reaction Design and Discovery (ICReDD), Hokkaido University, Sapporo, Hokkaido, Japan. D3 Center, The University of Osaka, Osaka, Japan. Center for Advanced Intelligence Project, RIKEN, Tokyo, Japan. \texttt{oki@icredd.hokudai.ac.jp}},\quad Ryosuke Sato\thanks{School of Science and Engineering, Chuo University, Tokyo, Japan. \texttt{rsato636@g.chuo-u.ac.jp}}}
\date{\today}
\usepackage[top=1.1truein,bottom=1.1truein,left=0.95truein,right=0.95truein]{geometry}
\usepackage{amsthm}
\usepackage{amsmath}
\usepackage{amssymb}
\usepackage{hyperref}
\usepackage{enumitem}
\usepackage{amsfonts}
\usepackage{bm}
\usepackage{algorithm}
\usepackage{tikz}

\usepackage{mathtools}
\usepackage{standalone}
\usepackage{multirow}
\usetikzlibrary{positioning,arrows.meta, decorations.pathreplacing}
\usetikzlibrary{calc}
\usepackage{subcaption}
\usepackage{algorithmic}
\usepackage{comment}
\usepackage[authoryear,round]{natbib}
\usepackage{booktabs}
\usepackage{cleveref}
\usepackage{xspace}

\theoremstyle{plain}
\newtheorem{theorem}{Theorem}[section]      
\newtheorem{proposition}[theorem]{Proposition}
\newtheorem{lemma}[theorem]{Lemma}
\newtheorem{corollary}[theorem]{Corollary}
\newtheorem{fact}[theorem]{Fact}

\theoremstyle{definition}
\newtheorem{definition}[theorem]{Definition}
\newtheorem{remark}[theorem]{Remark}

\newtheorem{assumption}[theorem]{Assumption}

\newtheorem{example}[theorem]{Example}

\providecommand{\mathbb}[1]{\mathbf{#1}}
\newcommand{\R}{\mathbb{R}}

\newcommand{\Z}{\mathbb{Z}}

\newcommand{\Tc}{\check{T}}

\newcommand{\rhot}{\tilde{\rho}}
\newcommand{\rhoc}{\check{\rho}}
\newcommand{\rhoh}{\hat{\rho}}
\newcommand{\Dc}{\check{D}}

\newcommand{\cL}{\mathcal{L}}
\newcommand{\cY}{\mathcal{Y}}
\newcommand{\LSC}{\textbf{(LSC)}\xspace}

\DeclarePairedDelimiter{\set}{\{}{\}}
\DeclarePairedDelimiterX{\Set}[2]{\{}{\}}{\,#1\mathclose{}\nonscript\;\delimsize|\nonscript\;\mathopen{}#2\,}
\DeclarePairedDelimiter{\abs}{|}{|}

\let\angle\relax
\DeclarePairedDelimiter{\angle}{\langle}{\rangle}

\providecommand{\texorpdfstring}[2]{#1}
\providecommand{\url}[1]{\texttt{#1}}

\DeclareMathOperator{\supp}{supp}
\DeclareMathOperator*{\argmax}{arg\,max}
\DeclareMathOperator*{\argmin}{arg\,min}

\DeclareMathOperator{\dom}{dom}
\DeclareMathOperator*{\minimize}{minimize}
\DeclareMathOperator*{\maximize}{maximize}

\newcommand{\zeros}{\mathbf{0}}
\newcommand{\ones}{\mathbf{1}}
\newcommand{\Mn}{M\texorpdfstring{${}^\natural$}{-natural}}

\newcommand{\Ln}{L\texorpdfstring{${}^\natural$}{-natural}}

\newcommand{\cA}{\mathcal{A}}

\begin{document}
\maketitle

\begin{abstract}
We study equilibrium computation in multi-unit trading networks with valuations satisfying generalized full substitutes and agent- and transaction-specific piecewise-linear transfer functions. These transfers model platform fees, taxes, subsidies, and membership discounts in online marketplaces. Our ascending procedure raises prices at different rates and reaches the minimum equilibrium price vector in finitely many iterations, each implementable in strongly polynomial time under our oracle and arithmetic models. An equilibrium allocation can then be recovered efficiently. The method extends primal-dual techniques developed for ascending auctions to handle complementarities between purchases and sales through a submodular flow formulation. In single-seller markets with substitutes and complements, our framework also extends the double-track auctions of Sun and Yang (2009) and Shioura and Yang (2015) to bidder- and good-specific~transfers.
\end{abstract}


\maketitle

\vspace{1cm}
\setcounter{tocdepth}{2} 

\section{Introduction}

Online marketplaces facilitate trade among individuals and firms
across countries and regions.
Agents may face different tax and subsidy regimes, while
platform fees and payment rules may vary with membership
status, transaction conditions, or price ranges.
These \emph{frictions} can create a gap between posted prices
and the amounts buyers actually pay or sellers receive.
The same price change may also affect monetary transfers differently
across agents and transactions.
Such differences shape trading incentives and complicate
price discovery, raising a natural question: how can price
adjustment clear markets under heterogeneous~frictions?

Consider a distributor buying from manufacturers and selling to retailers in an online wholesale marketplace.
Its purchasing decisions depend on the revenue it can earn
from reselling those goods, net of fees.
\emph{Trading networks}~\citep{HKNOW2013} model these interdependent
decisions, with directed edges representing potential
bilateral trades.
When purchases and sales face different fee schedules,
uniform price increases may raise purchasing costs
faster than net sales revenue.
Existing work establishes conditions for equilibrium
existence under frictions~\citep{FJJT2019} and for
a lattice structure of equilibrium prices~\citep{S2022}.
These results do not provide an algorithm
for computing an equilibrium.

Price adjustment in trading networks must handle
\emph{complementarities} between purchases and sales.
Such complementarities also arise in single-seller markets.
A retail platform selling complementary goods may apply
different fees across categories and membership tiers.
For goods such as laptops and monitors,
fees in one category can affect demand in the other,
complicating price adjustment.
\citet{SY2009} and \citet{SY2015} develop price-adjustment
procedures for markets with substitutes within each of two
groups and complements across groups, albeit without frictions.

Our goal is to develop a price-adjustment framework
for computing equilibria in trading networks under  
heterogeneous frictions, with applications to single-seller markets.
Without frictions, existing methods compute equilibria in
trading networks~\citep{CEV2021} and single-seller markets with substitutes
and complements~\citep{SY2009,SY2015} through welfare
maximization or its dual formulation.
Frictions can invalidate this characterization,
preventing a direct extension of these approaches~\citep{S2022}.
We therefore extend recent primal-dual techniques for ascending auctions~\citep{ENPRVV2025,OS2026} using \emph{minimum-cost
submodular flow}~\citep{EG1977,IMS2005}.

\subsection{Our Main Results}
We consider multi-unit trading networks with generalized
full-substitutes valuations and agent- and
transaction-specific piecewise-linear transfers.
Our framework also covers single-seller markets with substitutes
and complements.
It provides a positive answer to the question
of whether ascending price adjustment under frictions
extends beyond the settings of \citet{BEF2024} and
\citet{OS2026}.
Our main results are summarized as follows:
\begin{enumerate}
\item \textbf{An ascending procedure for multi-unit trading networks.}
We develop an ascending procedure that raises prices
at different rates to accommodate heterogeneous transfers.
Starting from zero prices, the procedure reaches the minimum
equilibrium price vector in finitely many iterations,
establishing equilibrium existence constructively.
At each iteration, it determines both the set of trades
whose prices increase and a direction specifying their
rates of increase while keeping this set locally stable.
The set and the direction can be computed in~strongly
polynomial time under our oracle and arithmetic~models.

\item \textbf{Applications to single-seller markets.}
Our trading-network formulation encompasses the model of
\citet{OS2026} as a special case, recovering its minimum
equilibrium good prices by projection, preserving buyer payments.
We also extend the celebrated double-track auctions~\citep{SY2009,SY2015}
for markets with substitutes and complements to bidder- and
good-specific~transfers.

\item \textbf{Equilibrium allocation recovery.}
Given equilibrium prices, we compute an equilibrium allocation
maximizing total trade volume in strongly polynomial time
under our oracle and arithmetic models.
To compute this allocation, we use minimum-cost
submodular flow, as in the price-adjustment~procedure.
\end{enumerate}

\paragraph{Our Techniques.}
We extend the primal-dual framework of \citet{OS2026}
from single-seller markets to trading networks, where agents
may both buy and sell and purchases and sales can be 
complementary.
Using the twist operator~\citep{IT2015,ISST2015}, we reverse sale coordinates
to obtain generalized-polymatroid demand sets that may
contain negative coordinates and whose rank functions
need not be nonnegative or monotone.
This prevents a direct application of their
polymatroid-based formulation.
We therefore use minimum-cost submodular flow, while 
retaining their lexicographic construction of price-update
directions from componentwise minimal optimal dual solutions.
The resulting dual effective-domain constraints yield the
sign conditions needed to prove local stability without
assuming monotonicity of rank functions.
The submodular-flow formulation makes these directions
computable in strongly polynomial time under our oracle
and arithmetic~models.

We further generalize the formulation by introducing a
(signed) market-clearing vector $\sigma$, which equals
the supply vector $s$ in the market studied by
\citet{OS2026} and zero in trading networks.
For markets with substitutes and complements, we assign~opposite signs to the supplies of the two groups.
Together with a reversal of price coordinates, this 
enables our procedure to also implement double-track
adjustment~\citep{SY2009,SY2015} under heterogeneous~frictions.

\subsection{Other Related Work}
Trading networks extend the supply-chain framework
of \citet{O2008} to more general network
structures~\citep{HKNOW2013}.
A central condition in this framework is full
substitutability, for which \citet{HKNOW2019} provide
equivalent characterizations for agents who both buy and sell.
In multi-unit settings, \citet{IT2015} study stability
in supply chain networks, while \citet{ISST2015} analyze stability 
and competitive equilibria in trading networks.
More recent work examines the complexity of stability
with discrete contracts~\citep{FJST2023}, as well as
convergence to equilibrium~\citep{LEKBKGG2025} and  stability and fairness~\citep{FGLT2026} in decentralized networks.
Much of this literature assumes no frictions,
whereas we focus on computing competitive equilibria
under heterogeneous frictions.

Ascending auctions and related price-adjustment procedures
for computing equilibrium prices have a long history,
primarily in markets with substitutable goods and
transferable utility;
see, e.g., \citep{KC1982, DGS1986,GS2000,A2006}.
Such price adjustment has also been studied for markets with
substitutes and complements~\citep{SY2009,SY2015} 
and for markets with imperfectly transferable
utility~\citep{A1992,DHW2015,BEF2024,OS2026}.
Our framework provides a unified approach that accommodates both complementarities and
heterogeneous piecewise-linear~transfers.

Submodular optimization~\citep{F2005} and discrete convex 
analysis~\citep{M1998,M2003} form the basis of our analysis.
Submodular flows were introduced by \citet{EG1977}, 
and \citet{IMS2005} give a strongly polynomial algorithm
for the minimum-cost problem.
\citet{FT1988} provide a unified treatment of generalized
polymatroids and submodular flows, while \citet{MS1999}
connect generalized polymatroids to discrete convexity. For a survey of discrete convex analysis in economics, see \citet{M2022}.

\section{Preliminaries}
\label{Preliminaries}
In this section, we review the discrete convex structures
and submodular flow techniques used in subsequent sections
to analyze demand sets and compute price-update directions.

Let $\R_+$ and $\Z_+$ denote the sets of nonnegative reals
and nonnegative integers, respectively.
For a finite set $E$ and a vector $x\in\R^E$,
write $x_e$ for its $e$-th component,
$x|_X$ for its restriction to $X\subseteq E$,
and $x(X)\coloneqq\sum_{e\in X}x_e$.
Let $\angle{x,y}$ denote the inner product of $x$~and~$y$.
Vector inequalities are componentwise.
The positive and negative supports of $x$ are
$\supp_+(x)\coloneqq\Set{e\in E}{x_e>0}$ and
$\supp_-(x)\coloneqq\supp_+(-x)$.
For $X\subseteq E$, $\chi_{\scriptscriptstyle X}\in\set{0,1}^E$
is the characteristic vector of $X$; we write
$\chi_{\scriptstyle e}\coloneqq\chi_{\scriptscriptstyle\{e\}}$
and $\chi_{\scriptstyle 0}\coloneqq\zeros$.
For $a\le b$ in $\Z^E$, let
$[a,b]_\Z\coloneqq\Set{x\in\Z^E}{a\le x\le b}$.
For an extended-valued function $f$, its \emph{effective domain}
$\dom f$ is the set of points where $f$ is finite.

We use the exact real-arithmetic model: arithmetic operations, comparisons, and evaluations of
$\log$ and $\exp$ have unit cost.

\subsection{Submodular Systems and \texorpdfstring{M-convex Sets}{M-convex Sets}}
\label{Mconvex}
An (integral) set function $\rho:2^E\to\Z$ is \emph{submodular} if
\[
\rho(X)+\rho(Y)\geq \rho(X\cup Y)+\rho(X\cap Y)
\qquad (X,Y\subseteq E),
\]
and \emph{supermodular} if $-\rho$ is submodular.
Let $\rho\colon 2^E\to\Z$ be a submodular function with
$\rho(\emptyset)=0$, not necessarily monotone or nonnegative.
The pair $(2^E,\rho)$ is called an (integral) \emph{submodular system}~\citep{F2005},\footnote{For monotone and nonnegative $\rho$, the polyhedron
$P(\rho)\cap\R_+^E$ is called a \emph{polymatroid}.}
and its associated (integral) \emph{submodular polyhedron} and \emph{base polyhedron} are defined by
\begin{align*}
P(\rho) &\coloneqq \Set{y \in \R^E}{y(S) \le \rho(S) \text{ for all } S \subseteq E}, \\
B(\rho) &\coloneqq \Set{y \in P(\rho)}{y(E) = \rho(E)},
\end{align*}
respectively.
$B(\rho)$ is bounded and consists precisely of the
componentwise maximal elements of $P(\rho)$. 
We use the following basic operations on submodular systems:
\begin{itemize}
    \item \textbf{Restriction:} for $S \subseteq E$, $\rho|_S(S') \coloneqq \rho(S')$ for $S' \subseteq S$.
    \item \textbf{Contraction:} for $X \subseteq E$, $\rho^X(S) \coloneqq \rho(S \cup X) - \rho(X)$ for $S \subseteq E \setminus X$. $B(\rho^X)$ is the projection onto $E\setminus X$ of $\Set{y\in B(\rho)}{y(X)=\rho(X)}$.
    \item \textbf{Direct sum:} for $\rho_1,\rho_2$ on disjoint ground sets
    $E_1,E_2$, respectively, define $(\rho_1\oplus\rho_2)(S)\coloneqq \rho_1(S\cap  E_1)+\rho_2(S\cap E_2)$.
\end{itemize}
Note that these operations preserve submodularity. 

We next recall the associated notions of discrete convexity.
\begin{definition}[\citet{M2003}]
\label{def:m_convex}
A nonempty subset $B \subseteq \Z^E$ is {\rm M}-convex if for any $x, y \in B$ and $e \in \supp_{+}(x-y)$, there exists $f \in \supp_{-}(x-y)$ such that $x - \chi_{e} + \chi_{f} \in B$ and $y + \chi_{e} - \chi_{f} \in B$.
It is {\rm\Mn}-convex if the same condition holds with $f \in \supp_{-}(x-y) \cup \set{0}$.
\end{definition}
The integer points of $P(\rho)$ and $B(\rho)$ are \Mn- and M-convex sets, respectively.
Conversely, every bounded M-convex set $B$ equals $B(\rho)\cap\Z^E$ with $\rho(S)=\max_{y\in B}y(S)$.
The function $\rho$ is submodular and is called the rank function of $B$.
Every \Mn-convex set $Q$ consists of the integer points
of an integral generalized polymatroid~\citep{MS1999,FT1988}.
For bounded $Q$, define its lower and upper rank functions by
$\mu(S)\coloneqq\min_{y\in Q} y(S)$ and
$\rho(S)\coloneqq\max_{y\in Q} y(S)$.
Then, $\mu$ is supermodular, $\rho$ is submodular, and
$Q$ has the representation
\[
Q=\Set{y\in\Z^E}{
        \mu(S)\le y(S)\le\rho(S)
        \text{ for all }S\subseteq E}.
\]

To accommodate negative coordinates and rank functions
that need not be monotone or nonnegative, we extend prior
polymatroid-based frameworks~\citep{ENPRVV2025,OS2026}
using a submodular flow formulation.

\subsection{Submodular Flow with Lexicographic Costs}
\label{subsec:lex-flow}

We introduce the submodular flow  problem and its 
dual~\citep{FT1988,F2005}.
We restrict the presentation to bipartite networks with all arcs
directed from one part to the other, as this suffices for our purposes, while the results below hold for general~networks.

Let $V = V^+ \cup V^-$ be a finite set partitioned into two parts, and let $A \subseteq V^+ \times V^-$ be a set of arcs, each directed from $V^+$ to $V^-$.
Let $\delta_v^+$ and $\delta_v^-$ denote the sets of
outgoing and incoming arcs at $v$, respectively.
Vectors in $\R^A$ and $\R^V$ are called a \emph{flow} and a \emph{potential}, respectively.
The \emph{boundary} of a flow $f\in\R^A$ at $v \in V$ is ${(\partial f)}_v \coloneqq f(\delta^+_v)$ for $v \in V^+$ and ${(\partial f)}_v \coloneqq -f(\delta^-_v)$ for $v \in V^-$.
For a potential $\pi\in\R^V$ and $a = (u,v) \in A$, let ${(\delta \pi)}_a \coloneqq \pi_u - \pi_v$.

Let $(2^{V^+}, \rho^+)$ and $(2^{V^-}, \rho^-)$ be submodular systems and $c \in \R^A$ a cost vector.
The \emph{submodular flow problem} and its Legendre-Fenchel dual are the following problems:
\begin{align}
    &\begin{array}{ll}
        \displaystyle\minimize_{f \in \R^A} & \angle{c, f} \\
        \text{subject to} & \partial f \in \mathcal{P},
    \end{array}
    &&
    \begin{array}{ll}
        \displaystyle\maximize_{\pi \in \R^V} & \displaystyle\inf_{y \in \mathcal{P}} \angle{\pi, y} \\
        \text{subject to} & \delta \pi = c,
    \end{array}\label{prob:flow}
\end{align}
where
\begin{equation}\label{eq:boundary-polyhedron-B}
    \mathcal{P} = \Set{y \in \R^V}{y|_{V^+} \in B(\rho^+),\ -y|_{V^-} \in B(\rho^-)} .
\end{equation}
By the integrality of $\rho^\pm$, there exists an integral optimal flow, provided that an optimal flow exists.
The following shows that optimal flows and
minimal optimal potentials can be computed~efficiently.

\begin{theorem}[{\citet{IMS2005}, \citet{M2003}}]\label{fact:sf-algorithms}
Given value oracles of $\rho^\pm$, an integral optimal flow of~\eqref{prob:flow} and the componentwise minimal optimal potential, if they exist, can be computed in strongly polynomial time.
\end{theorem}

We consider the problem~\eqref{prob:flow} with lexicographic cost $c=c^0+\gamma c^1$, where $c^0, c^1 \in \R^A$ and $\gamma$ is a formal infinitesimal.
We call~\eqref{prob:flow} with $c=c^0$ the \emph{primary problem}.
The lexicographic problem minimizes the secondary cost $c^1$
among the optimal solutions for the primary cost $c^0$.
Let $\pi^0$ be an optimal potential of the primary problem.~Define
\[
    \mathcal{P}^0
    \coloneqq
    \argmin_{y \in \mathcal{P}} \angle{\pi^0,y}.
\]
By the optimality conditions for convex flows~\citep[Section~9]{M2003},
a feasible flow $f$ is optimal for the primary problem
if and only if $\partial f \in \mathcal{P}^0$.
The lexicographic problem therefore reduces to the following
secondary problem and its dual:
\begin{align}
    &\begin{array}{ll}
        \displaystyle\minimize_{f \in \R^A} & \angle{c^1, f} \\
        \text{subject to} & \partial f \in \mathcal{P}^0,
    \end{array}
    &&
    \begin{array}{ll}
        \displaystyle\maximize_{\pi \in \R^V} & \displaystyle\inf_{y \in \mathcal{P}^0} \angle{\pi, y} \\
        \text{subject to} & \delta \pi = c^1.
    \end{array}\label{prob:flow-lex2}
\end{align}
We use the following decomposition of optimal potentials
to construct the update direction in
Section~\ref{subsec:lexico-dual}:
\begin{theorem}[\citet{OS2026}, Theorem~7.2 in the full version]\label{lexico_characterization}
Let $\pi^0$ be an optimal potential of~\eqref{prob:flow} with $c=c^0$, and $\pi^1$ be an optimal potential of~\eqref{prob:flow-lex2}.
Then, $\pi^0 + \gamma \pi^1$ is an optimal potential of~\eqref{prob:flow} with $c=c^0+\gamma c^1$. 
Conversely, every optimal potential of~\eqref{prob:flow} with $c=c^0+\gamma c^1$
has this form for some optimal potentials $\pi^0$ and $\pi^1$.
\end{theorem}

To solve the lexicographic problem, we represent each $a+\gamma b$ by the pair $(a,b)\in\R^2$
and compare these pairs lexicographically:
\[
    (a,b)>(a',b')
    \quad\Longleftrightarrow\quad
    a>a' \ \text{or}\ (a=a' \ \text{and}\ b>b').
\]
Addition, subtraction, and multiplication by real constants are performed componentwise.
\begin{remark}
\label{lexicographic_computation}
The algorithms underlying Theorem~\ref{fact:sf-algorithms} use only addition, subtraction, multiplication by real constants, and comparison on cost values.
They can therefore operate directly on lexicographic pairs
without assigning a numerical value to $\gamma$.
\end{remark}

\begin{remark}
We can consider a version of submodular flow obtained by replacing the base~polyhedra $B(\rho^\pm)$ in~\eqref{eq:boundary-polyhedron-B} with the submodular polyhedra $P(\rho^\pm)$, i.e., the boundary polyhedron $\mathcal{P}$~is
\begin{equation}\label{eq:boundary-polyhedron-P}
    \mathcal{P} = \Set{y \in \R^V}{y|_{V^+} \in P(\rho^+),\ -y|_{V^-} \in P(\rho^-)}.
\end{equation}
The two problems with~\eqref{eq:boundary-polyhedron-B} and~\eqref{eq:boundary-polyhedron-P} can be reduced to each other in strongly polynomial time~\citep[Section~5.2]{F2005}, and Theorems~\ref{fact:sf-algorithms} and~\ref{lexico_characterization} remain valid for the problem with~\eqref{eq:boundary-polyhedron-P}. 
\end{remark}

\section{Our Model and Demand Structures}
In this section, we introduce our multi-unit trading-network
model with heterogeneous transfers and establish the demand
properties underlying our price-adjustment procedure.
All proofs in this section are given in Appendix~\ref{proof_sec3}.
\subsection{Our Model}
Consider a loopless trading network with a set of agents $N =\{1, \dots, n\}\; (n \geq 2)$ and a finite set $E$ of directed edges representing possible bilateral trades, with $m \coloneqq |E|$.
Each trade $e \in E$ is directed from its seller to its buyer.
The \emph{capacity} vector $c \in \Z_{+}^E$ specifies the maximum number of \emph{indivisible} units tradable on each edge.
For each agent $i \in N$, let $\delta_+(i)$ and $\delta_-(i)$ denote the sets of outgoing edges (trades on which $i$ sells) and incoming edges (trades on which $i$ buys), respectively, and let $\delta(i) \coloneqq \delta_+(i) \cup \delta_-(i)$.
We write $\cA\coloneqq\Set{(i,e)}{i\in N,\ e\in\delta(i)}$
for the set of agent-trade~incidences.

\begin{remark}
\label{zero_extension}
In the rest of this paper, we identify vectors indexed by $\delta(i)$
with vectors indexed by $E$ by setting all coordinates
outside $\delta(i)$ to zero.
This convention also applies to sets of~such~vectors.
\end{remark}

An \emph{outcome} consists of a \emph{trade vector} $z \in [\zeros, c]_{\Z}$ and a \emph{price vector} $p \in \R_+^E$, where $z_e$ is the number of units traded on $e$.
Let $z_i\coloneqq z|_{\delta(i)}$ be the bundle of trades
executed by agent $i$.
Each $i \in N$ has a valuation function
$v_i\colon [\zeros,c|_{\delta(i)}]_{\Z}\to\R\cup\{-\infty\}$
with nonempty effective domain, where $-\infty$ encodes hard trading constraints, such as the seller's inventory constraint.

To model \emph{frictions} between prices and monetary transfers, we introduce transfer functions, which are assumed to be publicly known.\footnote{This is a natural assumption since frictions typically
arise from observable factors such as platform fees, tax rates,
or tariffs, which are common knowledge.}
For each $(i,e)\in\cA$, let $q_{ie}:\R_+\to\R_+$ specify
agent $i$'s monetary transfer \emph{per unit} of trade $e$: if $e \in \delta_-(i)$, buyer $i$ \emph{pays} $q_{ie}(p_e) z_e$; if $e \in \delta_+(i)$, seller $i$ \emph{receives} $q_{ie}(p_e) z_e$. 
The posted price $p_e$ is common to the buyer and seller, 
while their per-unit transfers may differ. 
We impose the following assumptions on these functions: 
\begin{assumption}\label{ass:transfers}
Each transfer function $q_{ie}$ is continuous, strictly increasing,
and piecewise linear with finitely many pieces, and satisfies
$q_{ie}(0)=0$. We write $q'_{ie}$ for its right derivative.
\end{assumption}

At prices $p \in \R_+^E$, define agent $i$'s utility and demand set by
\begin{align*}
    u_i(z_i; p) &\coloneqq v_i(z_i) - \sum_{e\in \delta_-(i)} q_{ie}(p_e)z_e + \sum_{e\in \delta_+(i)} q_{ie}(p_e)z_e,\\
    D_i(p) &\coloneqq \argmax_{z_i\in [\zeros, c|_{\delta(i)}]_{\Z}} u_i(z_i;p).
\end{align*}
Our goal is to compute a competitive equilibrium, defined as~follows.
\begin{definition}
\label{Competitive}
A \emph{competitive equilibrium} is a pair $(z,p)$ with
$z\in[\zeros,c]_{\Z}$ and $p\in\R_+^E$ such that
$z_i\in D_i(p)$ for every $i\in N$.
If such a $z$ exists for $p$, we call $p$ an
\emph{equilibrium price vector}.
\end{definition}
With multi-unit trades represented as sets of unit contracts,
competitive equilibrium outcomes are also trail-stable
in the resulting market~\citep[Theorem~2]{FJJT2019}.

To formulate our substitutability assumption, we adopt a common
sign convention for purchases and sales.
Following \citet{IT2015} and \citet{ISST2015}, we define the \emph{twist operator}
$\tau_i:\Z^{\delta(i)}\to\Z^{\delta(i)}$~by
\[
    (\tau_i(z_i))_e \coloneqq
    \begin{cases}
        z_e  & \text{if } e \in \delta_-(i) \text{ (buying)}, \\
        -z_e & \text{if } e \in \delta_+(i) \text{ (selling)}.
    \end{cases}
\]

Let $\cY_i\coloneqq\tau_i([\zeros,c|_{\delta(i)}]_{\Z})$.
Define the \emph{twisted valuation} on $\cY_i$ by
$\tilde v_i(y_i)\coloneqq v_i(\tau_i(y_i))$.
The function $\tilde v_i$ is \emph{\Mn-concave} if, for any
$x,y\in\dom\tilde v_i$ and $e\in\supp_+(x-y)$, there exists
$f\in\supp_-(x-y)\cup\set{0}$ such that
$x-\chi_e+\chi_f,\ y+\chi_e-\chi_f\in\cY_i$ and
\[
\tilde v_i(x)+\tilde v_i(y)
\le \tilde v_i(x-\chi_e+\chi_f)
   +\tilde v_i(y+\chi_e-\chi_f).
\]

\citet{ISST2015} extend full substitutability in trading
networks~\citep{HK2012,HKNOW2013} to multi-unit trades,
introducing the generalized full-substitutes (GFS) condition.
GFS coincides with full substitutability for unit-capacity trades.
They show that GFS is equivalent to \Mn-concavity
of the twisted valuation. 
We impose GFS in this equivalent form and use properties
of \Mn-concave functions in our analysis.
\begin{assumption}
\label{ass:valuation}
For every $i\in N$, the valuation $v_i$
(i) is normalized, i.e., $v_i(\zeros)=0$;
(ii) is \emph{monotone in purchases}:
$v_i(z_i)\leq v_i(z'_i)$ whenever $z_i\leq z'_i$  and
$z_e=z'_e$ for all $e\in\delta_+(i)$; and
(iii) satisfies \emph{generalized full substitutes} (GFS),
equivalently, $\tilde v_i$ is \Mn-concave.
\end{assumption}
\begin{remark}
\label{rem:monotone}
\citet{OS2026} assume that buyers' valuations are
nondecreasing, whereas we require monotonicity only in
purchase coordinates.
Sellers may incur production costs, so $v_i$ may decrease
as sales increase and take negative values.
Monotonicity in purchases is only used in
Lemma~\ref{no_underdemanded}.
\end{remark}

Condition~(iii) also implies the laws of aggregate demand
and supply~\citep{ISST2015}.
We can therefore apply Schlegel's lattice
theorem~\citep[Theorem~1(i)]{S2022}
by expanding each trade into unit contracts. 

\begin{corollary}
\label{cor:equilibrium-lattice}
For any equilibrium price vectors $p$ and $r$,
their componentwise minimum $p\wedge r$ and maximum
$p\vee r$ are also equilibrium price vectors.
\end{corollary}

\citet{FJJT2019} established equilibrium existence under
full substitutability and bounded compensating variations.
Our assumptions do not guarantee the latter condition.
Instead, our ascending mechanism constructively establishes
the existence of an equilibrium and computes the minimum
equilibrium price vector.

\subsection{Demand Structures and Equilibrium Prices}
\label{subsec:twist}
We study the properties of agent $i$'s \emph{twisted demand set} at price $p$,
defined by $T_i(p)\coloneqq\tau_i(D_i(p))$.
Define agent $i$'s \emph{twisted utility} by
\begin{equation}
\label{eq:twisted-utility}
    \tilde u_i(y_i;p)
    \coloneqq
    u_i(\tau_i(y_i);p)
    =
    \tilde v_i(y_i)
    - \sum_{e \in \delta(i)} q_{ie}(p_e)y_e.
\end{equation}
The twist gives all agents a common utility representation,
encoding sales as negative quantities.
The twisted demand set is thus given~by
\[
    T_i(p)
    =
    \argmax_{y_i \in \cY_i} \tilde u_i(y_i;p).
\]
Let $\Tc_i(p)$ denote the componentwise minimal
elements of $T_i(p)$.
The following fact follows from
\citet[Theorem~5.1]{ISST2015} and
\citet[Theorem~4.5 and Remarks~4.1--4.2]{M2022}.
\begin{fact}
\label{prop:twisted_m_convex}
Under Assumption~\ref{ass:valuation}~(iii), for every $p\in\R_+^E$,
$T_i(p)$ is an \Mn-convex set and $\Tc_i(p)$ is an M-convex set.
\end{fact}
By M-convexity, all elements of $\Tc_i(p)$ have the same
total sum.
Moreover, for every $y_i\in T_i(p)$, there exists
$x_i\in\Tc_i(p)$ with $x_i\leq y_i$.
Therefore, these minimal elements are precisely
the bundles with minimum total sum:
\[
\Tc_i(p)=\argmin_{y_i\in T_i(p)} y_i(E).
\]

Using the zero-extension convention in
Remark~\ref{zero_extension}, let $\mu_i(\cdot;p)$ and
$\rhoh_i(\cdot;p)$ denote the lower and upper rank functions
of the twisted demand set $T_i(p)$, respectively,
as defined in Section~\ref{Mconvex}.
We write $\mu\coloneqq\sum_{i\in N}\mu_i$ and
$\rhoh\coloneqq\sum_{i\in N}\rhoh_i$.

To characterize equilibrium prices,
we introduce a \emph{market-clearing vector} $\sigma\in\Z^E$.
Market clearing in twisted coordinates is expressed as
$\sum_{i\in N}y_i=\sigma$, where $\sigma$ specifies
the required net purchases on each trade.
In trading networks, purchases and sales cancel on each trade,
so $\sigma=\zeros$.
We allow $\sigma$ to have entries of either sign
to accommodate the setting of~Section~\ref{sec:sy}.
Given $\sigma$, a set of trades $X\subseteq E$
is \emph{over-demanded} at $p$ if $\mu(X;p)>\sigma(X)$,
and \emph{under-demanded} if $\rhoh(X;p)<\sigma(X)$.
The following lemma characterizes equilibrium prices
by the absence of such~sets:
\begin{lemma}
\label{price_characterization}
For any $p\in\R_+^E$, there exist bundles
$y_i\in T_i(p)$ for all $i\in N$ satisfying
$\sum_{i\in N}y_i=\sigma$
if and only if no set is over-demanded or under-demanded at $p$.
\end{lemma}

The algorithm in Section~\ref{Mechanism} raises the prices of trades
in the unique minimal maximizer $X^*(p)$ of
$\mu(\cdot;p)-\sigma(\cdot)$.\footnote{Since
$\mu(\cdot;p)-\sigma(\cdot)$ is supermodular, its maximizers
are closed under union and intersection and thus have a
unique inclusionwise minimal element.}
We next express this maximization problem in terms of
the rank functions of the minimal twisted demand sets.
Define the rank function of $\check T_i(p)$ by
\begin{align}
\label{rank}
    \rhoc_i(S;p)
    \coloneqq
    \max_{y_i\in\Tc_i(p)}y_i(S)
    \quad(S\subseteq E),
\end{align}
which is integer-valued and submodular in $S$
for each fixed $p$, but need not be monotone or nonnegative.
We write $\rhoc\coloneqq\sum_{i\in N}\rhoc_i$.
Since the minimum defining $\mu_i$ is attained on
$\Tc_i(p)$~\citep[Fact~A.3]{OS2026}, whose elements
have total sum $\rhoc_i(E;p)$, we have
$\mu_i(X;p)=\rhoc_i(E;p)-\rhoc_i(E\setminus X;p)$.
Summing over agents~yields
\begin{equation}
\label{mu_rho}
    \mu(X;p)=\rhoc(E;p)-\rhoc(E\setminus X;p).
\end{equation}
By~\eqref{mu_rho}, we have the following relation:
\begin{equation}
\label{argmin}
    \argmax_{X\subseteq E}
    \bigl\{\mu(X;p)-\sigma(X)\bigr\}
    =
    \argmin_{X\subseteq E}
    \bigl\{\sigma(X)+\rhoc(E\setminus X;p)\bigr\}.
\end{equation}
To analyze how sufficiently small price increases affect
the rank functions, we use the following decomposition
into contributions from the updated coordinates
and their complement:
\begin{lemma}
\label{lem:perturbation_decomposition}
For any $p,d\in\R_+^E$ with $X=\supp_+(d)$,
there exists $\delta>0$ such that, for all
$\varepsilon\in(0,\delta)$, $i\in N$, and $S\subseteq E$,
\begin{align}
\rhoh_i(S;p+\varepsilon d)
&=\rhoh_i(S\setminus X;p)
+\rhoc_i(S\cap X;p+\varepsilon d),\\
\rhoc_i(S;p+\varepsilon d)
&=\rhoc_i(S\setminus X;p)
+\rhoc_i(S\cap X;p+\varepsilon d).
\end{align}
\end{lemma}

To apply the proof of \citet[Proposition~3.8~(ii)]{OS2026}
in twisted coordinates, it suffices to establish the following
demand inclusions, since the remaining argument does not
require nonnegative demand vectors.
The lemma is proved in Appendix~\ref{proof-decomposition}.
\begin{lemma}
\label{subset}
For any $p,d\in\R_+^E$, there exists $\delta>0$ such that,
for every $i\in N$ and all $\varepsilon\in(0,\delta)$,
$T_i(p+\varepsilon d)\subseteq T_i(p)$ and 
$\Tc_i(p+\varepsilon d)\subseteq\Tc_i(p)$.
\end{lemma}

\section{Our Mechanism}
\label{Mechanism}
In this section, we present a framework for computing competitive
equilibria via directional price updates.
These updates allow prices to rise at different rates,
accommodating heterogeneous transfers.
Proofs omitted from this section are given in Appendix~B.

\subsection{Overview}
\begin{algorithm}[t]
\caption{\textsc{Directional Ascending Price Adjustment}}
\label{algo1}
\begin{algorithmic}[1]
\STATE Set $p\leftarrow\zeros$.
\LOOP
    \STATE Compute a direction $d\in\R_+^E$ satisfying \LSC
at $p$, together with its support $X=\supp_+(d)$.
    \IF{$X=\emptyset$}
        \STATE \textbf{return} $p$.
    \ENDIF
    \STATE Compute $\lambda=\lambda_{\mathrm{evt}}(p,d)$
    using~\eqref{event}.
    \STATE Set $p\leftarrow p+\lambda d$.
\ENDLOOP
\end{algorithmic}
\end{algorithm}

Our mechanism, presented in Algorithm~\ref{algo1}, starts from
zero prices and, at each nonterminal iteration, computes
a direction and a step size for updating the prices.
Following \citet{OS2026}, we require each direction
$d\in\R_+^E$ to satisfy the following \emph{local stability
condition} \LSC at price $p$, where $X=\supp_+(d)$:
\[
\boxed{
    \exists\,\delta>0 \quad
    X=X^*(p+\varepsilon d)
    \quad\text{for all }\varepsilon\in[0,\delta).
}
\tag{\textbf{LSC}}
\]
Taking $\varepsilon=0$ gives $X=X^*(p)$. 
The following example illustrates why prices may need
to rise at different rates to satisfy \LSC.
\begin{example}
\label{ex:intermediary}
Consider a chain of three agents: seller $S$,
intermediary $I$, and buyer $B$.
The trades $a:S\to I$ and $b:I\to B$ each have unit capacity.
Seller $S$ has unit production cost $1$, and buyer $B$
has unit value $4$.
Intermediary $I$ can resell only after purchasing:
$v_I(z_a,z_b)=0$ if $z_b\le z_a$ and $-\infty$ otherwise.
All transfers are identity functions except that $I$
pays $2p_a$ per purchased unit.
Its resale profit is therefore $p_b-2p_a$.
This market falls outside the one-sided settings
\citep{BEF2024,OS2026}, but is still covered by our framework.

At $p=(0,0)$, we have $\Tc_I(p)=\{(0,0),(1,-1)\}$.
Its rank function is nonmonotone:
$\rhoc_I(\{a\};p)=1>0=\rhoc_I(\{a,b\};p)$.
Since $\rhoc(\emptyset;p)=0$ and
$\rhoc(\{a\};p)=\rhoc(\{b\};p)=\rhoc(\{a,b\};p)=1$,
we have $X^*(p)=\{a,b\}$.
At $p=(\varepsilon,\varepsilon)$ for sufficiently small
$\varepsilon>0$, resale is unprofitable and $I$ does not trade.
Thus, $X^*(p)=\{b\}$, violating \LSC.

The direction $d^*=(1/2,1)$ instead keeps resale profit
at zero, preserving $X^*=\{a,b\}$ and satisfying \LSC.
At step size $2$, the seller $S$ becomes willing to trade,
yielding the minimum equilibrium price vector $p^*=(1,2)$
with both trades executed.
\end{example}

To compute update directions satisfying \LSC, 
we assume access to the following oracles for each agent $i \in N$:

\begin{itemize}
\item \textbf{Demand Oracle.}
Given $p\in\R_+^E$, return some bundle $y_i\in T_i(p)$.

\item \textbf{Exchange Capacity Oracle.}
Given $p\in\R_+^E$, $y_i\in T_i(p)$, and distinct
$e,f\in\delta(i)\cup\{0\}$, return the maximum $\alpha \in \Z_+$ such that $y_i-\alpha\chi_e+\alpha\chi_f\in T_i(p)$.
\end{itemize}

The following theorem guarantees efficient computation
of a direction satisfying \LSC.
We prove it in \Cref{subsec:lexico-dual,subsec:solving}.

\begin{theorem}
\label{thm:direction-computable}
For any price $p\in\R_+^E$, a direction $d^*$ satisfying \LSC at $p$ can be computed in strongly polynomial time using the demand and exchange capacity oracles for all agents.
\end{theorem}
For computing the step size in Line~7 of Algorithm~\ref{algo1},
we additionally assume access to the following oracle
for each $i\in N$:

\begin{itemize}
\item\textbf{Demand-Change Oracle.}
Given $p,d\in\R_+^E$, return $\lambda_{\mathrm{dem},i}(p,d)$, defined as the infimum of all $\lambda>0$ such that $D_i(p+\lambda d)\ne D_i(p+td)$ for some $t\in(0,\lambda)$.
\end{itemize}
Using the (public) transfer functions $q_{ie}$, we also compute $\lambda_{\mathrm{tr}}(p,d)$, which is the infimum of $\lambda>0$ such that $q'_{ie}(p_e+\lambda d_e)\ne q'_{ie}(p_e)$ for some $(i,e) \in \mathcal{A}$.
In both definitions, $\inf\emptyset=+\infty$.
Then, the step size to the next event in Line~7 is given by
\begin{equation}
\label{event}
    \lambda_{\mathrm{evt}}(p,d)
    \coloneqq
    \min\set[\Big]{
        \min_{i\in N}\lambda_{\mathrm{dem},i}(p,d),
        \lambda_{\mathrm{tr}}(p,d)
    }.
\end{equation}
Nonterminal step sizes are positive and finite;
see Appendix~\ref{Convergence}.

\begin{remark}
Our framework applies even without access to numerical valuations,
provided agents can answer the above queries.
Our oracle model also includes the value-oracle model as a special case:
a value oracle returning $v_i(\bm z_i)$ for a given bundle
$\bm z_i$ suffices to implement all three oracles in polynomial time.
\end{remark}
The next theorem establishes correctness and finite termination.
\begin{theorem}
\label{thm:correctness_algo1}
Whenever Algorithm~\ref{algo1} terminates, it outputs
the minimum equilibrium price vector. 
Using the algorithm of Theorem~\ref{thm:direction-computable},
it terminates in finitely many iterations.
\end{theorem}
Minimality and finite termination are proved in
Appendices~\ref{minimality} and~\ref{Convergence}, respectively.
Here, we show that the output is an equilibrium price vector
using the following invariant.

\begin{lemma}
\label{no_underdemanded}
Throughout Algorithm~\ref{algo1}, no set is under-demanded.
\end{lemma}

At termination, no set is over-demanded by the stopping
condition in Line~4 or under-demanded by Lemma~\ref{no_underdemanded}.
Thus, Lemma~\ref{price_characterization} implies that
the output $p^*$ is an equilibrium price vector.
At $p^*$, an equilibrium allocation can be computed in
strongly polynomial time using the demand and exchange
capacity oracles; see Section~\ref{Allocation}.

To prove the invariant, we use \LSC and the following
characterization. 
This characterization is also used in the proof of
Lemma~\ref{d_weighted}.

\begin{lemma}
\label{any_direction}
Let $X^*=X^*(p)$ and $d\in\R_+^E$ satisfy $\supp_+(d)=X^*$.
For sufficiently small $\varepsilon>0$, the following are equivalent:
\begin{enumerate}[label={\rm(\alph*)},leftmargin=2em]
    \item The direction $d$ satisfies \LSC at $p$.
    \item $\rhoc(X;p+\varepsilon d)>\sigma(X)$
    for every nonempty $X\subseteq X^*$.
\end{enumerate}
\end{lemma}

\begin{proof}[Proof of Lemma~\ref{no_underdemanded}]
We proceed by induction over iterations.
At $p=\zeros$, monotonicity in purchases and $q_{ie}(0)=0$
allow us to choose $y_i^0\in T_i(\zeros)$ with
$y_{ie}^0=c_e$ for all $e\in\delta_-(i)$.
Since $y_{ie}^0\geq-c_e$ for $e\in\delta_+(i)$, 
no set is initially under-demanded by 
\[
\rhoh(X;\zeros)
\geq \sum_{i\in N}y_i^0(X\cap\delta(i))
\geq \sum_{e\in X}(c_e-c_e)
=0=\sigma(X)
\quad(X\subseteq E).
\]

Suppose no set is under-demanded at the current price $p$.
Let $d$ and $\lambda$ be the update direction and step size,
respectively, and let $p'=p+\lambda d$ be the next price vector.
Write $X^*=X^*(p)=\supp_+(d)$.
Since $d$ satisfies \LSC, Lemmas~\ref{lem:perturbation_decomposition} 
and~\ref{any_direction}~(b) give
\begin{align*}
\rhoh(X_1\cup X_2;p+td)
=\rhoc(X_1;p+td)+\rhoh(X_2;p)
\geq \sigma(X_1)+\sigma(X_2)
=\sigma(X_1\cup X_2)
\end{align*}
for any $X_1\subseteq X^*$, $X_2\subseteq E\setminus X^*$,
and sufficiently small $t>0$.
Since demand sets are constant on $0<t<\lambda$,
the inequality holds throughout.
Finally, finiteness of the feasible sets and continuity of
utilities give 
$T_i(p+(\lambda-\varepsilon)d)\subseteq T_i(p')$
for all $i$ and sufficiently small $\varepsilon>0$.
Thus, $\rhoh(X;p')\geq\sigma(X)$ for all $X\subseteq E$,
completing the induction.
\end{proof}

\subsection{Constructing Directions with LSC}\label{subsec:lexico-dual}

This section constructs a direction satisfying \LSC at a given price $p \in \R_+^E$.
We write $X^*=X^*(p)$ and abbreviate
$P(\rhoc_i(\cdot;p))$ as $P(\rhoc_i)$ for $i\in N$.
We identify each $x_i\in\R^{\delta(i)}$ with its zero extension to $E$.
For the primal-dual formulation below, define the weights~by
\[
    w_{ie}\coloneqq 1-\gamma\log q'_{ie}(p_e)
    \qquad((i,e)\in\mathcal A),
\]
where $\gamma$ is the formal infinitesimal used for
lexicographic comparisons in Section~\ref{subsec:lex-flow}.
We consider the following primal-dual pair:
\begin{align*}
(P)\quad &
\max_{(x_1,\dots,x_n)}
\Set*{
    \sum_{(i,e)\in\mathcal A}w_{ie}x_{ie}
}{
    x_i\in P(\rhoc_i)\ (\forall i),\
    \sum_{i\in N}x_i\leq\sigma
},
\notag\\
(D)\quad &
\min_{z\in\dom\cL}\ \cL(z),
\end{align*}
where the dual objective $\cL\colon \R_+^E \to \R \cup \set{+\infty}$ is defined by
\[
    \cL(z)\coloneqq
    \sum_{i\in N}
    \sup_{x_i\in P(\rhoc_i)}
    \set[\Bigg]{
        \sum_{e\in\delta(i)}(w_{ie}-z_e)x_{ie}
    }
    +\sum_{e\in E}\sigma_ez_e.
\]
In the following, we derive the primal-dual relation between $(P)$ and $(D)$ via reduction to the submodular flow problem.
This reduction also enables us to apply Theorems~\ref{fact:sf-algorithms} and~\ref{lexico_characterization}.

Let $E'\coloneqq\Set{e'}{e\in E}$ be a copy of $E$,
and define the directed bipartite graph $D=(V,A)$ by
\[
    V\coloneqq\mathcal A\cup E',
    \qquad
    A\coloneqq
    \Set{\bigl((i,e),e'\bigr)}{(i,e)\in\mathcal A},
\]
where each incidence node $(i,e)$ has a single outgoing arc
to its trade node $e'$.
Take $V^+\coloneqq\mathcal A$ with
$\rho^+\coloneqq\bigoplus_{i\in N}\rhoc_i$
and $V^-\coloneqq E'$ with $\rho^-(S)\coloneqq\sigma(S)$.
Here, each $\rhoc_i$ is restricted to $\delta(i)$
and indexed by its incidence nodes,
and we identify $E'$ with $E$.
Write $y_i\coloneqq(y_{(i,e)})_{e\in\delta(i)}$.
Let $\mathcal{P}$ be the polyhedron~\eqref{eq:boundary-polyhedron-P} defined by $\rho^\pm$, i.e.,
\[
    \mathcal P
    =
    \Set{y\in\R^V}{
        y_i\in P(\rhoc_i)\ (i\in N),\
        -y_{e'}\leq\sigma_e\ (e\in E)
    }.
\]
Set $c_a\coloneqq-w_{ie}$ for each arc $a=((i,e),e')$,
and identify a flow $f\in\R^A$ with $x\in\R^{\mathcal A}$
by $f_{((i,e),e')}=x_{ie}$.
Then $\angle{c,f}=-\angle{w,x}$.

The identities $(\partial f)_{(i,e)}=x_{ie}$ and
$(\partial f)_{e'}=-\sum_{i:\,e\in\delta(i)}x_{ie}$
show that the boundary constraints are equivalent to
$x_i\in P(\rhoc_i)$ and
$\sum_{i:\,e\in\delta(i)}x_{ie}\leq\sigma_e$.
After negating the objective, $(P)$ is equivalent to~\eqref{prob:flow} with $c=c^0+\gamma c^1$, where
$c^0=-\ones$ and
$c^1=(\log q'_{ie}(p_e))_{(i,e)\in\mathcal A}$.

We next show that negating the flow dual objective yields $(D)$.
Let
$\pi=\bigoplus_{i\in N}\pi_i^+\oplus\pi^-\in\R^V$
be a potential on $D$, where
$\pi_i^+=(\pi_{ie}^+)_{e\in\delta(i)}$ and
$\pi^-=(\pi_e^-)_{e\in E}$.
Since $\mathcal P$ is the direct sum over $V^+$ and $V^-$,
the dual of \eqref{prob:flow} is
\[
\begin{array}{ll}
    \displaystyle\maximize_{\pi\in\R^V}
    &
    \displaystyle
    \sum_{i\in N}
    \inf_{x_i\in P(\rhoc_i)}
    \set[\Bigg]{
        \sum_{e\in\delta(i)}\pi_{ie}^+x_{ie}
    }
    +\sum_{e\in E}\inf_{-\sigma_e\leq y}\pi_e^-y
    \\[1ex]
    \text{subject to}
    &
    \pi_{ie}^+-\pi_e^-=-w_{ie}
    \quad((i,e)\in\mathcal A).
\end{array}
\]
The following conditions characterize
the dual objective's effective domain.
Since $P(\rhoc_i)$ is down-closed and unbounded below,
the term for agent $i$ is finite if and only if
$\pi_{ie}^+\leq0$ for all $e\in\delta(i)$,
and the infimum is attained on $B(\rhoc_i)$.
The term for trade $e$ is finite if and only if
$\pi_e^-\geq0$, and then equals $-\sigma_e\pi_e^-$,
regardless of the sign of $\sigma_e$.

Set $z\coloneqq\pi^-$ so that $\pi_{ie}^+=z_e-w_{ie}$.
The finiteness conditions give
\begin{align}
    \dom\cL
    =
    \Set{z\in\R_+^E}{
        z_e\leq w_{ie}\text{ for all }(i,e)\in\mathcal A
    }.\label{eq:dom-L}
\end{align}
On this domain, the flow dual objective is $-\cL(z)$.
Thus, $(P)$ and $(D)$ form a primal-dual pair
with equal optimal values.
By Theorem~\ref{lexico_characterization}, we obtain the following
characterization of the minimal optimal solution of $(D)$.
Its proof is deferred to Section~\ref{subsec:convex-dual-proof}.
\begin{theorem}\label{convex-dual}
The dual problem $(D)$ has a unique componentwise minimal optimal solution
of the form $z^*=\chi_{X^*}+\gamma t$ for some $t\in\R^E$.
\end{theorem}

We now construct a direction from the minimal optimal solution
$z^*$ of $(D)$.
With $t$ from Theorem~\ref{convex-dual},
define $d^*\in\R_+^E$ by $d_e^*\coloneqq\exp(t_e)$
for $e\in X^*$ and $d_e^*\coloneqq0$ otherwise.
At each iteration, our implementation of Algorithm~\ref{algo1}
uses $d^*$ constructed at the current price as its update direction.

\subsection{Proof of Theorem~\ref{thm:direction-computable}}\label{subsec:solving}
Here, we show that the direction $d^*$ constructed above
satisfies \LSC and can be computed in strongly polynomial time.

We first show \LSC.
Let $z^* = \chi_{X^*}+\gamma t$ be the minimal optimal solution to $(D)$.
By construction, we have $\supp_+(d^*)=X^*$.
By $z^* \in \dom \mathcal{L}$ and~\eqref{eq:dom-L}, $z_e^*\le w_{ie}$ holds for $(i,e)\in\mathcal A$.
Specifically, for $(i,e)\in\mathcal A$ with $e\in X^*$, we have
\begin{equation}\label{eq:c0}
    \kappa_{ie}\coloneqq\log q'_{ie}(p_e)+t_e\le0.
\end{equation}
This bound leads to the following inequality, which allows us to derive the strict inequalities
in Lemma~\ref{any_direction}~(b).

\begin{lemma}\label{ineq_L}
Let $z^*$ be the minimal optimal solution of $(D)$
and $d^*$ the direction defined in Section~\ref{subsec:lexico-dual}.
For sufficiently small $\varepsilon>0$ and every nonempty
$X\subseteq X^*$, there exists $\tau>0$ such~that
$    \cL(z^*-\gamma\xi\chi_{\scriptscriptstyle X})
    \le \cL(z^*)+\gamma\xi
    \bigl(\rhoc(X;p+\varepsilon d^*)-\sigma(X)\bigr)
    \quad (\forall\xi\in(0,\tau)).$
\end{lemma}
To prove Lemma~\ref{ineq_L}, we first describe the
base-polyhedron decomposition induced by $d^*$.~Write
\[
A:=E\setminus X^*,
\qquad
\rhoc_i:=\rhoc_i(\cdot;p),
\qquad
\rhot_i:=\rhoc_i^A.
\]
We use $\oplus$ to denote the direct sum over disjoint
coordinate sets.
\begin{lemma}
\label{greedy_structure}
For all sufficiently small $\varepsilon>0$ and every $i\in N$,
\begin{equation}
\label{eq:ineq-L-base-decomposition}
B\bigl(\rhoc_i(\cdot;p+\varepsilon d^*)\bigr)
=
B(\rhoc_i|_A)\oplus
\argmax_{y_i\in B(\rhot_i)}
\sum_{e\in X^*\cap\delta(i)}(-\kappa_{ie})y_{ie}.
\end{equation}
\end{lemma}

\begin{proof}
For sufficiently small $\varepsilon>0$, local linearity
of the transfers implies that the price increase from $p$
to $p+\varepsilon d^*$ changes the net payments for
a bundle $y_i$ by
\[
\varepsilon
\sum_{e\in X^*\cap\delta(i)}
q'_{ie}(p_e)d_e^*y_{ie}.
\]
By Lemma~\ref{subset},
both $T_i(p+\varepsilon d^*)\subseteq T_i(p)$ and
$\Tc_i(p+\varepsilon d^*)\subseteq\Tc_i(p)$ hold.
Since all bundles in $T_i(p)$ yield the same utility at $p$,
the updated demand consists precisely of those minimizing
this payment increase.
The second inclusion ensures that the minimum is attained
in $\Tc_i(p)$.
Conversely, any minimizer in $\Tc_i(p)$ is demanded after
the update and remains componentwise minimal.
Thus, $\Tc_i(p+\varepsilon d^*)$ consists precisely of
the minimizers of the payment increase over~$\Tc_i(p)$.

Taking convex hulls gives the corresponding minimizing
face of $B(\rhoc_i)$.
Since the coefficients are zero on $A$ and positive on $X^*$,
the greedy rule decomposes this face into $B(\rhoc_i|_A)$
and a minimizing face of $B(\rhot_i)$.
On $B(\rhot_i)$, replacing the positive coefficients
$q'_{ie}(p_e)d_e^*$ by their logarithms
$\kappa_{ie}=\log q'_{ie}(p_e)+t_e$
preserves their weak ordering and hence the minimizing face.
\end{proof}

\begin{proof}[Proof of Lemma~\ref{ineq_L}]
Fix $\varepsilon>0$ sufficiently small for
Lemma~\ref{greedy_structure},
and let $\emptyset\neq X\subseteq X^*$.
For $\xi\geq0$, define 
$z^\xi:=z^*-\gamma\xi\chi_{\scriptscriptstyle X}$ and 
$\kappa^\xi_{ie}:=
\kappa_{ie}-\xi(\chi_{\scriptscriptstyle X})_e$
for any $e\in X^*\cap\delta(i)$. 
The perturbed point $z^\xi$ belongs to $\dom\cL$:
its primary components on $X$ remain one, preserving
nonnegativity, and $z^\xi\leq z^*$ preserves the upper bounds.
Thus, all coefficients $w_{ie}-z^\xi_e$ are~lexicographically nonnegative, and each supremum
in $\cL(z^\xi)$ is attained on $B(\rhoc_i)$.
For each $i\in N$,~define
\[
h_i(\xi):=
\max_{y_i\in B(\rhot_i)}
\sum_{e\in X^*\cap\delta(i)}(-\kappa^\xi_{ie})y_{ie}.
\]

In each maximization defining $\cL(z^\xi)$,
the primary coefficients are one on $A$ and zero on $X^*$.
The primary objective therefore maximizes $x_i(A)$.
The greedy decomposition separates the optimization over
$B(\rhoc_i|_A)$ from that over $B(\rhot_i)$.
The former contributes a term independent of $\xi$,
while the latter contributes $\gamma h_i(\xi)$.
The linear term $\sum_{e\in E}\sigma_e z^\xi_e$
changes by $-\gamma\xi\sigma(X)$.
Subtracting the value at $\xi=0$ gives
\begin{equation}
\label{eq:ineq-L-difference}
\cL(z^\xi)-\cL(z^*)
=
-\gamma\xi\sigma(X)
+\gamma\sum_{i\in N}\bigl(h_i(\xi)-h_i(0)\bigr).
\end{equation}

Choose $\tau>0$ sufficiently small that, for every $i\in N$
and $\xi\in(0,\tau)$, the perturbation from $-\kappa_i$
to $-\kappa^\xi_i$ preserves every strict coefficient
ordering on $X^*\cap\delta(i)$.
A nonincreasing ordering for $-\kappa^\xi_i$ is then
also valid for $-\kappa_i$.
The greedy solution $y_i^\xi$ for this ordering
therefore attains both $h_i(\xi)$ and $h_i(0)$.
By~\eqref{eq:ineq-L-base-decomposition},
it extends to a point of
$B(\rhoc_i(\cdot;p+\varepsilon d^*))$.
Evaluating both objectives at this common optimizer yields
$h_i(\xi)-h_i(0)
=
\xi y_i^\xi(X)
\leq
\xi\rhoc_i(X;p+\varepsilon d^*)$. 
Substituting into~\eqref{eq:ineq-L-difference} gives
$\cL(z^\xi)-\cL(z^*)
\leq
\gamma\xi\bigl(\rhoc(X;p+\varepsilon d^*)-\sigma(X)\bigr)$.
\end{proof}

We now establish \LSC for~$d^*$.
\begin{lemma}\label{d_weighted}
The direction $d^*$ satisfies \LSC at $p$.
\end{lemma}

\begin{proof}
Fix $\varepsilon>0$ sufficiently small for
Lemmas~\ref{any_direction} and~\ref{ineq_L}.
For any nonempty $X\subseteq X^*$, take $\tau>0$
from Lemma~\ref{ineq_L}.
For $\xi\in(0,\tau)$, we have
$z^*-\gamma\xi\chi_{\scriptscriptstyle X}\in\dom\cL$.
Since $z^*$ is the componentwise minimal optimal solution of $(D)$, we have 
$\cL(z^*-\gamma\xi\chi_{\scriptscriptstyle X})>\cL(z^*)$.
Lemma~\ref{ineq_L} therefore gives
$\rhoc(X;p+\varepsilon d^*)>\sigma(X)$.
Since $X$ was arbitrary and $\supp_+(d^*)=X^*$,
Lemma~\ref{any_direction} shows that $d^*$ satisfies~\LSC.
\end{proof}

The submodular flow algorithm in Theorem~\ref{fact:sf-algorithms}
requires value oracles for $\rho^+$ and $\rho^-$,
which define $\mathcal P$.
Since $\rho^-(S)=\sigma(S)$ is directly computable,
it suffices to construct a value oracle for $\rho^+$.
These values can be computed efficiently using the demand
and exchange capacity oracles;
see Section~\ref{app:boundary-oracles} for the proof.

\begin{lemma}\label{lem:boundary-oracles}
    Let $\rho^+$ be the function defined in Section~\ref{subsec:lexico-dual}.
    Given the demand and exchange capacity oracles for all agents, the value $\rho^+(S)$ for $S \subseteq \mathcal{A}$ can be computed in strongly polynomial time.
\end{lemma}

We are now ready to prove Theorem~\ref{thm:direction-computable}.

\begin{proof}[Proof of Theorem~\ref{thm:direction-computable}]
The flow formulation in Section~\ref{subsec:lexico-dual}
has $O(\abs{\mathcal A}+\abs{E})$ nodes and arcs.
By Lemma~\ref{lem:boundary-oracles} and
Theorem~\ref{fact:sf-algorithms},~the minimal optimal potential $\pi=\pi^0+\gamma\pi^1$
can be computed in strongly polynomial time
using lexicographic pairs (Remark~\ref{lexicographic_computation}).
Since the correspondence with dual solutions preserves
componentwise order, $z^*=\pi|_{E'}$ is the minimal
optimal solution of $(D)$.~
By Theorem~\ref{convex-dual}, we recover
$X^*=\Set{e\in E}{\pi^0_{e'}=1}$
and set $d_e^*=~\exp(\pi^1_{e'})$~on $X^*$
and zero otherwise.
Lemma~\ref{d_weighted} shows that $d^*$ satisfies \LSC.
\end{proof}

\begin{remark}\label{rem:simpler}
The overall proof strategy is similar to that of \citet{OS2026}, but the technical details differ. In their setting, demand vectors and the submodular functions
defining our flow formulation are nonnegative,
allowing a polymatroid-intersection formulation. In our setting, neither property is guaranteed due to the twist. We instead exploit the constraints imposed by $\dom \mathcal{L}$,
which yield~\eqref{eq:c0} throughout $X^*$ and 
enable the preceding~analysis.
\end{remark}

\subsection{Computing Equilibrium Allocations}\label{Allocation}
This section describes an algorithm for computing an equilibrium allocation $y = {(y_i)}_{i\in N} \in \R^\cA$ from an equilibrium price $p^* \in \R_+^E$, i.e., $y_i\in T_i(p^*)$ for every $i\in N$ and $\sum_{i\in N}y_i=\sigma$. 
Such a profile exists by Lemma~\ref{price_characterization}.
We focus on allocations that maximize total trade volume.
In our single-seller application in Section~\ref{sec:single-seller-comparison}, this ensures
that all supply, including zero-price goods, is allocated.

\begin{proposition}
\label{prop:allocation}
Given an equilibrium price $p^*$ and access to the demand and exchange capacity oracles for all agents, an equilibrium allocation maximizing total trade volume can be computed in strongly polynomial time.
\end{proposition}

\begin{proof}
Consider the directed bipartite graph $D=(V^+ \cup V^-, A)$ in Section~\ref{subsec:lexico-dual},
and identify a flow $f\in\R^A$ on $D$ with $y\in\R^{\cA}$ by $f_{((i,e),e')}=y_{ie}$.
We consider boundary polytopes
\begin{align}
    \mathcal{P}^+ &\coloneqq \textstyle\bigoplus_{i \in N} \operatorname{conv}(T_i(p^*)) \cap \Set{y \in \R^\mathcal{A}}{y(\mathcal{A}) = \sigma(E)},\label{eq:Pplus}\\
    \mathcal{P}^- &\coloneqq  \set{\sigma},\notag
\end{align}
on $V^+ = \mathcal{A}$ and $V^- = E'$, respectively, where $\operatorname{conv}$ denotes the convex hull.
Since each $T_i(p^*)$ is \Mn-convex, the fixed-sum slice
$\mathcal P^+$ and the singleton $\mathcal P^-$ are 
base polyhedra of submodular functions~$\rho^\pm$.
Thus, integral feasible flows correspond to
equilibrium allocations in twisted coordinates.
We show in Section~\ref{app:boundary-oracles} how to compute
$\rho^+(S)$ in strongly polynomial time;
$\rho^-(S)$ is directly~computable.

Set $c^1\coloneqq\zeros$ and define $c^0_{((i,e),e')}$
to be $-1$ for $e\in\delta_-(i)$ and $0$ for $e\in\delta_+(i)$.
For the corresponding trade vector $z$, counting each trade
on its buyer's arc gives $\angle{c^0,f}=-\sum_{e\in E}z_e$.
Thus, minimizing flow cost maximizes total trade volume.
The feasible region is nonempty and bounded.
By Theorem~\ref{fact:sf-algorithms}, an integral optimal flow
can be computed in strongly polynomial time.
\end{proof}

\begin{remark}
\label{rem:allocation-sy}
The same construction applies to the network in
Section~\ref{sec:sy}, with $\cA=N\times M$ and
$\sigma=\tau(s)$, where $\tau$ is the twist operator defined there
and $s$ is the supply vector.
Feasibility requires all supply to be sold.
Any integral feasible flow gives an equilibrium allocation,
and we may therefore take $c^0=c^1=\zeros$.
\end{remark}

\section{Applications to Single-Seller Markets}
\label{sec:applications}

In this section, we apply our framework to single-seller markets
with a fixed supply of goods.
We first recover the minimum equilibrium price vector in the
market model of \citet{OS2026} via a transaction-price projection.
We then extend the double-track auction \citep{SY2009,SY2015} 
to bidder- and good-specific transfer~functions.

We consider a market with a finite set $M$ of goods and a seller
supplying $s_j$ units of each good $j\in M$ to a set
$N=\{1,2,\ldots,n\}$ of bidders, where each $s_j$
is a positive integer and $n\ge2$.
Each bidder $i$ has a coordinatewise nondecreasing valuation
$\phi_i:[0,s]_{\mathbb Z}\to\mathbb R$.
Without loss of generality, normalize each valuation so that
$\phi_i(\zeros)=0$.
At prices $p\in\mathbb R_+^M$, 
define bidder $i$'s utility and demand set by
\begin{align*}
    u_i(x_i;p)
    &\coloneqq
    \phi_i(x_i)-\sum_{j\in M}q_{ij}(p_j)x_{ij},\\
    D_i(p)
    &\coloneqq
    \argmax_{x_i\in[0,s]_{\mathbb Z}}u_i(x_i;p),
\end{align*}
where each $q_{ij}$ satisfies 
Assumption~\ref{ass:transfers}.
A competitive equilibrium consists of prices $p$
and bundles $x_i\in D_i(p)$ satisfying $\sum_i x_i=s$.

\subsection{The Single-Seller Correspondence}
\label{sec:single-seller-comparison}
We apply our trading-network framework to the market of
\citet{OS2026} and compare the resulting prices with its
minimum equilibrium prices.
Assume that each valuation $\phi_i$ is \Mn-concave.
The network has agents $N\cup\{0\}$, where $0$ is the seller.
For each $(i,j)\in N\times M$, we introduce a trade $e_{ij}$
from $0$ to $i$ with capacity $s_j$.
The bidders' valuations and payment functions remain unchanged.
The seller's transfer functions are identities.
Its valuation is zero if
$\sum_{i\in N}z_{e_{ij}}\le s_j$ for every $j\in M$,
and $-\infty$ otherwise.

Transaction prices for the same good may differ across bidders.
At equilibrium, however, profit maximization requires the seller to
trade only at the highest transaction price for each good.
Let $p^\star$ be the price vector returned by
Algorithm~\ref{algo1} on this network, and let $\bar p$
denote the minimum equilibrium price vector of the original~market.
\begin{lemma}
\label{thm:single-seller-minimum-prices}
For each $j\in M$,
$\bar p_j=\max_{i\in N}p_{e_{ij}}^\star$.
Moreover, let $x^\star$ be an equilibrium allocation in the
trading network at $p^\star$ maximizing total trade volume,
and identify it with the bidders' bundles.
Then, $(x^\star,\bar p)$ is a competitive equilibrium
of the original market.
\end{lemma}

\begin{proof}
Let $x^\star$ be an allocation as in the statement and define
$P\in\mathbb R_+^M$ by
$P_j\coloneqq\max_{i\in N}p_{e_{ij}}^\star$.
The seller's demanded bundle yields the maximum
revenue $\sum_j s_jP_j$.
Hence, every executed trade in good $j$ has price $P_j$,
and $s_j$ units are sold whenever~$P_j>0$.

If good $j$ is not fully allocated, then
$p_{e_{ij}}^\star=P_j=0$ for every $i$.
Assigning one additional unit of $j$ to any bidder is feasible
and leaves all payments unchanged.
By monotonicity, the enlarged bundle remains demanded.
This yields another equilibrium at $p^\star$ with greater
total trade volume, contradicting the choice of $x^\star$.
Thus, $\sum_i x_i^\star=s$.
At good prices $P$, each bidder pays the same for $x_i^\star$
and weakly more for any alternative bundle.
Therefore, $(x^\star,P)$ is a competitive equilibrium 
of the original market, and $\bar p\le P$.

For the opposite direction, consider an equilibrium
of the original market at $\bar p$.
Assigning price $\bar p_j$ to every trade $e_{ij}$
and retaining~the allocation gives an equilibrium
in the trading network:
bidder utilities are unchanged, and selling all supply
maximizes seller revenue.
Since $p^\star$ is the minimum equilibrium price~vector
in the trading network, $p_{e_{ij}}^\star\le\bar p_j$
for all $i,j$.
Hence, $P\le\bar p$, proving $\bar p=P$.
\end{proof}

\subsection{Substitutes and Complements}
\label{sec:sy}

We now extend the substitutability assumption to allow
complementarities.
Fix a partition $M=S_1\mathbin{\dot\cup}S_2$.
Here, we adapt the twist operator to this goods partition:
${(\tau(x))}_j\coloneqq x_j$ for $j\in S_1$ and
${(\tau(x))}_j\coloneqq-x_j$ for $j\in S_2$.
The twisted valuation is
$(\phi_i\circ\tau)(y_i)\coloneqq\phi_i(\tau(y_i))$
on $\tau([0,s]_{\Z})$.
We assume that each $\phi_i$ satisfies \emph{generalized
gross substitutes and complements} (GSC):
goods are substitutes within each group and
complements across groups.
Equivalently, each twisted valuation $\phi_i\circ\tau$
is \Mn-concave~\citep{SY2015}.
When $S_2=\emptyset$, the twist is the identity and
the model reduces to the setting of
Section~\ref{sec:single-seller-comparison}.

For good-price vectors $p,r\in\mathbb R_+^M$,
define the order $\le_g$ by
\[
    p\le_g r
    \quad\Longleftrightarrow\quad
    p|_{S_1}\le r|_{S_1}
    \quad\text{and}\quad
    p|_{S_2}\ge r|_{S_2}.
\]
Corollary~\ref{cor:equilibrium-lattice}
also holds in this model, with meet and join taken
with respect to $\le_g$;
see Appendix~\ref{app:equilibrium-lattice-schlegel} for more details.

To apply Algorithm~\ref{algo1} to this market,
we reverse price coordinates on $S_2$, turning
double-track adjustment into an ascending process.
Choose $\bar p_j>0$ for each $j\in S_2$ satisfying
\[
    q_{ij}(\bar p_j)>\phi_i(s)-\phi_i(\zeros)
    \qquad (i\in N,\ j\in S_2).
\]
Unbounded transfers ensure such bounds exist.\footnote{To choose $\bar p$ via demand queries, start with
$r=1$ and double $r$ until each demand oracle returns
$\zeros$ at the original prices $r\ones$.
The search terminates since $\zeros$~is uniquely demanded
for sufficiently large $r$.
Its optimality yields
$\phi_i(s)-\phi_i(\zeros)\leq\sum_{k\in M}q_{ik}(r)s_k$.
It therefore suffices to choose $\bar p_j$ such that
$q_{ij}(\bar p_j)>\sum_{k\in M}q_{ik}(r)s_k$
for all $i\in N$.}
Since no bidder demands good $j$ when $p_j\ge\bar p_j$,
every equilibrium satisfies $p_j<\bar p_j$ on $S_2$.
We therefore restrict attention to $p\in\R_+^M$ with
$p_j\le\bar p_j$ for all $j\in S_2$, and define
the transformed prices by
$\pi_j\coloneqq p_j$ for $j\in S_1$ and
$\pi_j\coloneqq\bar p_j-p_j$ for $j\in S_2$.
Increasing $\pi$ raises prices on $S_1$ and lowers them
on $S_2$, as in the double-track auction \citep{SY2009,SY2015}. 
The componentwise order on $\pi$ corresponds to
$\le_g$ on the original~prices.

Set $y_i\coloneqq\tau(x_i)$ and normalize
the transfers at zero by defining
\begin{align*}
    \tilde\phi_i(y_i)
    &\coloneqq
    \phi_i(\tau(y_i))
    +\sum_{j\in S_2}q_{ij}(\bar p_j)y_{ij},\\
    Q_{ij}(\pi_j)
    &\coloneqq
    \begin{cases}
        q_{ij}(\pi_j), & j\in S_1,\\
        q_{ij}(\bar p_j)-q_{ij}(\bar p_j-\pi_j),
        & j\in S_2.
    \end{cases}
\end{align*}
In these coordinates, each bidder's utility is given by
\[
    u_i(x_i;p)
    =
    \tilde\phi_i(y_i)
    -\sum_{j\in M}Q_{ij}(\pi_j)y_{ij}.
\]
The added linear term preserves \Mn-concavity,
and hence the transformed demand set
$T_i(\pi)\coloneqq\tau(D_i(p))$ is \Mn-convex.
Each $Q_{ij}$ satisfies the conditions of
Assumption~\ref{ass:transfers} on its domain.

Market clearing becomes $\sum_i y_i=\sigma$,
where $\sigma\coloneqq\tau(s)$ has entries $s_j$ on $S_1$
and $-s_j$ on $S_2$.
With this $\sigma$, Lemma~\ref{price_characterization}
characterizes equilibrium prices in the transformed coordinates.
Assuming access to the oracles of Section~\ref{Mechanism}
for the transformed demand sets, we then run our implementation
of Algorithm~\ref{algo1} from $\pi=\zeros$.
This corresponds to initial prices $p_j=0$ on $S_1$
and $p_j=\bar p_j$ on $S_2$.

\begin{theorem}
\label{thm:sy-convergence}
The double-track mechanism computes the minimum
equilibrium price vector under $\le_g$ in finitely many
iterations.
Each update direction is computable in strongly polynomial
time under our oracle and arithmetic models.
\end{theorem}

This establishes equilibrium existence constructively.
Among all equilibria, the output minimizes prices in $S_1$
and maximizes prices in $S_2$.
An equilibrium allocation can be obtained by applying
the construction in Remark~\ref{rem:allocation-sy}
at these prices and undoing the twist.

\begin{proof}
For the analysis, extend each $Q_{ij}$ on $S_2$
linearly beyond $\bar p_j$ with a positive slope.
Define $\Tc_i(\pi)$ and $\rhoc_i(\cdot;\pi)$ as before.
The rank decomposition in 
Lemma~\ref{lem:perturbation_decomposition}
then applies in these coordinates.
Use the network of Section~\ref{subsec:lexico-dual}
with goods $M$ and incidences $N\times M$.
Take $\rho^+=\bigoplus_i\rhoc_i$
and $\rho^-(H)=\sigma(H)$ for $H\subseteq M$,
and set $w_{ij}\coloneqq1-\gamma\log Q'_{ij}(\pi_j)$
using right derivatives.
For $j\in S_2$ with $\pi_j<\bar p_j$,
$Q'_{ij}(\pi_j)$ equals the left derivative of $q_{ij}$
at $p_j=\bar p_j-\pi_j$.
The construction in the proof of
Theorem~\ref{thm:direction-computable}
therefore yields a direction satisfying \LSC
in strongly polynomial~time.

At $\pi=\zeros$, monotonicity and the choice of $\bar p$
ensure that each bidder has a demanded bundle containing
the full supply of every group-1 good and no group-2 goods.
The sum of these demanded bundles in twisted coordinates
is componentwise at least $\sigma$.
Thus, no set is initially under-demanded.
The induction step of Lemma~\ref{no_underdemanded}
preserves this invariant throughout.

We next show that the price path remains bounded.
For $j\in S_1$, choose $b_j$ such that
$q_{ij}(b_j)>\phi_i(s)-\phi_i(\zeros)$ for every $i$.
At $\pi_j=b_j$, every demanded bundle has $y_{ij}=0$.
For $j\in S_2$, at $\pi_j=\bar p_j$ we have $p_j=0$,
and monotonicity implies $y_{ij}=-s_j$
for every $y_i\in\Tc_i(\pi)$.

At either boundary, removing $j$ from a set $X$ containing it
changes $\sigma(X)+\rhoc(M\setminus X;\pi)$ by
$-s_j$ on $S_1$ and $-(|N|-1)s_j$ on $S_2$.
Since these changes are nonpositive, the inclusionwise
minimal minimizer $X^*(\pi)$ cannot contain $j$.
By \LSC and the choice of step size, the active set remains
the minimal minimizer before the next event.
An update increasing coordinate $j$ therefore cannot cross
the boundary, where $j$ must be inactive.
Thus, the price path is bounded.
Since each nonterminal update increases some coordinate
at a fixed positive rate, its step length is finite.

The argument in Appendix~\ref{Convergence} applies to the
transformed model and establishes finite termination.
At termination, Lemma~\ref{price_characterization}
guarantees an equilibrium.
Every equilibrium price vector of the original market
lies within the chosen bounds and therefore corresponds to
an equilibrium price vector of the transformed model.
The argument in Appendix~\ref{minimality} then establishes
componentwise minimality of the terminal prices.
In the original price coordinates, this yields the minimum
equilibrium price vector under $\le_g$.
\end{proof}

\section{Concluding Remarks}
\label{Concluding}
We developed an ascending-price procedure for
multi-unit trading networks with generalized full-substitutes
valuations and piecewise-linear transfers.
It computes the minimum equilibrium price vector
in finitely many iterations, with each direction
computable in strongly polynomial time under our oracle
and arithmetic models.
At these prices, an equilibrium allocation maximizing
total trade volume can be computed in strongly
polynomial time.
The method combines a twist of sale coordinates
with submodular flow to compute locally stable directions.
We also recover minimum good prices in the
single-seller model of \citet{OS2026} by projection while
preserving realized buyer payments, and extend double-track
adjustment~\citep{SY2009,SY2015}
to bidder- and good-specific transfers.

Future work could explore more general transfer functions,
in addition to seeking polynomial iteration bounds.
Extending beyond piecewise-linear transfers would require
stability and termination arguments that do not rely
on locally constant marginal rates.
Another direction is to identify economically natural
valuation representations enabling efficient implementation
of our~oracles.

\section*{Acknowledgements}
	Taihei Oki was supported by JST FOREST Grant Number JPMJFR232L, JST BOOST Grant Number JPMJBY25A6, and JSPS KAKENHI Grant Number JP22K17853.
    Ryosuke Sato was supported by JST ERATO Grant Number JPMJER2301 and JSPS KAKENHI Grant Number JP26K21174.

We used OpenAI's GPT and Anthropic's Claude to assist with developing research ideas, constructing and checking proofs, and drafting and revising the manuscript.
The authors take full responsibility for the content of this paper, including all mathematical claims and proofs.

\bibliographystyle{plainnat}
\bibliography{competitive}

\appendix

\section{Omitted Proofs in Section 3}
\label{proof_sec3}

\subsection{Lattice Properties}
\label{app:equilibrium-lattice-schlegel}

We establish the lattice properties for trading networks
and the transformed single-seller market of
Section~\ref{sec:sy} by reducing both to unit-trade markets.
This expansion is used only to establish the lattice
property; our algorithm works directly with the original
capacities.

\begin{proof}[Proof of Corollary~\ref{cor:equilibrium-lattice}]
Replace each edge $e$ by $c_e$ unit trades with the same
endpoints and transfer functions.
Extend the transfers continuously and strictly increasingly
to negative prices.

Represent each agent's choice by purchased units and
retained sale units.
For a bundle $z_i$, write $y_i=\tau_i(z_i)$.
On a sale edge $e$, the retained quantity is
$c_e-z_e=c_e+y_{ie}$.
The valuation on these nonnegative counts is a translation
of $\tilde v_i$ and remains \Mn-concave.
By \citet[Theorem~4.1]{ST2015}, it satisfies strong
substitutes: treating each unit as a distinct item yields
a gross-substitutes, and hence \Mn-concave, set valuation.
Utility equals this valuation minus the effective prices
of the selected units, plus a term depending only on prices.
The transfers preserve coordinatewise order and equality, and the
number of selected units equals net purchases plus a constant.
The gross-substitutes and aggregate-demand properties
of \Mn-concave set valuations
\citep[Theorems~4.1(ii) and~4.7(i)]{ST2015}
imply full substitutability and the laws of aggregate
demand and supply for the expanded trade demand.
These properties hold even when copies of the same trade
have different prices.

Let $\mathcal E$ and $\widehat{\mathcal E}$ denote the
equilibrium price sets of the original and expanded markets,
respectively.
Let $L$ assign each edge price to all its unit copies.
At prices $L(p)$, choosing copies to match each trade
quantity gives an expanded equilibrium; counting the
executed copies recovers an original equilibrium.
For $p\in\R_+^E$,
$p\in\mathcal E$ if and only if
$L(p)\in\widehat{\mathcal E}$.
By \citet[Theorem~1(i)]{S2022}, $\widehat{\mathcal E}$
is a lattice under coordinatewise minimum and maximum.
Since $L$ preserves these operations,
$\mathcal E$ has the same lattice property.
\end{proof}

For the setting of Section~\ref{sec:sy}, apply the same
unit-trade reduction to the transformed prices $\pi$,
valuations $\tilde\phi_i$, and transfers $Q_{ij}$.
To enforce fixed supply, introduce one auxiliary seller
per unit on $S_1$ and one auxiliary buyer per unit on $S_2$.
Each must execute exactly one trade with a bidder,
with zero valuation on feasible bundles and identity transfers.
These auxiliary agents satisfy full substitutability and
the laws of aggregate demand and supply.
After extending $Q_{ij}$ continuously and strictly increasingly
to real prices, the preceding argument establishes
these properties for the bidders' expanded~demands.

At prices common to all unit trades of each good,
the auxiliary agents are indifferent between bidders.
The same correspondence between equilibria therefore holds,
with the auxiliary constraints enforcing $\sum_i y_i=\sigma$.
Restricting to common prices and to
$\R_+^{S_1}\times\prod_{j\in S_2}[0,\bar p_j]$
preserves the lattice property.
Since all original equilibrium prices lie within these bounds,
transforming back gives the lattice property under $\le_g$.
Nonemptiness follows from Theorem~\ref{thm:sy-convergence}.

\subsection{Proof of Lemma~\ref{price_characterization}}
\begin{proof}
Let $T(p)\coloneqq\sum_{i\in N}T_i(p)$ be the Minkowski sum
of the twisted demand sets.
There exist bundles $y_i\in T_i(p)$ satisfying
$\sum_{i\in N}y_i=\sigma$ if and only if $\sigma\in T(p)$.
Since each $T_i(p)$ is \Mn-convex,
so is $T(p)$~\citep[Theorem~4.23]{M2003}.
Moreover, its convex hull $\mathrm{conv}(T(p))$ is an integral polyhedron
whose integer points are exactly the elements of $T(p)$
\citep[Section 4.7]{M2003}.
Since $\sigma$ is integral,
$\sigma \in T(p)$ is equivalent to $\sigma \in \mathrm{conv}(T(p))$. 
The convex hull is explicitly given~by 
$\mathrm{conv}(T(p))
=\Set{y\in\R^E}{\mu(X;p)\leq y(X)\leq \rhoh(X;p) \text{ for all }X\subseteq E}$.
Therefore, $\sigma\in T(p)$ is equivalent to
$\mu(X;p)\leq \sigma(X)\leq \rhoh(X;p)$
for all $X\subseteq E$.
These two inequalities exactly state that no set is over-demanded or
under-demanded at $p$, respectively.
\end{proof}

\subsection{Proof of Lemma~\ref{subset}}
\label{proof-decomposition}
\begin{proof}
Fix an agent $i\in N$.
We first show that $T_i(p+\varepsilon d)\subseteq T_i(p)$
for all sufficiently small $\varepsilon>0$.
Choose $y^*\in T_i(p)$.
For every
$y\in\operatorname{dom}\tilde v_i\setminus T_i(p)$,
we have $\tilde u_i(y^*;p)>\tilde u_i(y;p)$.
By continuity of the transfers and finiteness of
$\operatorname{dom}\tilde v_i$,
there exists $\delta_i>0$ such that
$\tilde u_i(y^*;p+\varepsilon d)>
\tilde u_i(y;p+\varepsilon d)$ 
for every such $y$ and every $\varepsilon\in(0,\delta_i)$.
Thus, no bundle outside $T_i(p)$ can be demanded at
$p+\varepsilon d$, proving the first inclusion.

Next, fix $\varepsilon\in(0,\delta_i)$ and let
$y\in\Tc_i(p+\varepsilon d)$. 
This implies $y\in T_i(p)$.
Suppose to the contrary that $y\notin\Tc_i(p)$.
Then, there exists $z\in T_i(p)$ with
$z\leq y$ and $z\neq y$.
Since $z$ and $y$ yield the same utility at $p$,
\begin{align*}
\tilde u_i(z;p+\varepsilon d)
-\tilde u_i(y;p+\varepsilon d)
=
\sum_{e\in\delta(i)}
\bigl(q_{ie}(p_e+\varepsilon d_e)-q_{ie}(p_e)\bigr)
(y_e-z_e)
\geq0,
\end{align*}
where the inequality follows from $z\leq y$ and monotonicity of the transfers.
Since $y$ is demanded at $p+\varepsilon d$,
this inequality implies that $z$ is also demanded there,
contradicting the componentwise minimality of $y$.
Hence, $\Tc_i(p+\varepsilon d)\subseteq\Tc_i(p)$. 
Taking $\delta\coloneqq\min_{i\in N}\delta_i>0$
establishes both inclusions simultaneously for all agents.
\end{proof}

To obtain Lemma~\ref{lem:perturbation_decomposition},
we apply the proof of
\citet[Proposition~3.8(ii) in the full version]{OS2026}
using the demand inclusions established above.
The remaining argument uses feasible exchanges, comparisons
of transfer increments, and submodularity of the rank functions.
These steps do not require nonnegative demand vectors,
and the argument therefore applies in twisted coordinates.

\section{Omitted Proofs in Section 4}
\label{proof_sec4}
\subsection{Proof of Minimality in Theorem~\ref{thm:correctness_algo1}}
\label{minimality}
We first show that, whenever Algorithm~\ref{algo1} terminates,
its output is no larger than any equilibrium price vector.
The proof of finite termination is deferred to
Appendix~\ref{Convergence}.
The following lemma describes how $\mu_i$ and $\rhoc_i$
change under componentwise price increases.
Its proof applies directly in twisted~coordinates.
\begin{lemma}[\citet{OS2026}, Proposition~3.6~(ii) in the full version]
\label{mu_rho_relation}
Let $p\leq p'$ and let $Y\subseteq\Set{e\in E}{p_e=p'_e}$.
For every $i\in N$, we have
$\mu_i(Y;p)\leq\mu_i(Y;p')$ and
$\mu_i(E\setminus Y;p)\geq\mu_i(E\setminus Y;p')$.
If the second inequality holds with equality, then
$\rhoc_i(E;p)=\rhoc_i(E;p')$ and
$\rhoc_i(Y;p)=\rhoc_i(Y;p')$.
\end{lemma}
\begin{proof}[Proof of minimality in Theorem~\ref{thm:correctness_algo1}]
Let $p^*$ be the price vector returned by Algorithm~\ref{algo1},
and let $\bar p$ be any equilibrium price vector.
Suppose, for a contradiction, that $p^*\not\leq\bar p$.
Since the algorithm starts from $\zeros\leq\bar p$,
consider the first update from 
$p^-$ to $p^+=p^-+\lambda d$
such that $p^+\not\leq\bar p$.
Then, $p^-\leq\bar p$.
Let $p$ be the last point on this segment satisfying
$p\leq\bar p$.
Thus, $p$ lies strictly before $p^+$, and
$p_e=\bar p_e$ for some $e\in\supp_+(d)$.

Set $X^*\coloneqq\supp_+(d)$.
By \LSC and the choice of step size,
$X^*$ is the inclusionwise minimal solution
of~\eqref{argmin} at $p$.
Define $X_1\coloneqq\Set{e\in X^*}{p_e=\bar p_e}\neq\emptyset$
and $X_2\coloneqq X^*\setminus X_1\subsetneq X^*$.
With the virtual direction $d'\coloneqq\chi_{X_2}$,
choose $\varepsilon>0$ sufficiently small that
Lemma~\ref{lem:perturbation_decomposition} applies and
$\widehat p\coloneqq p+\varepsilon d'\leq\bar p$. 
Since $\widehat p$ and $\bar p$ coincide on $X_1$,
summing Lemma~\ref{mu_rho_relation} over agents yields
\begin{align*}
\mu(X_1;\bar p)-\sigma(X_1)
&\geq \mu(X_1;\widehat p)-\sigma(X_1)\\
&=\rhoc(E;\widehat p)
  -\rhoc(E\setminus X_1;\widehat p)-\sigma(X_1)\\
&=\rhoc(E\setminus X_2;p)
  -\rhoc(E\setminus X^*;p)-\sigma(X_1)\\
&=\sigma(X_2)+\rhoc(E\setminus X_2;p)-\left(\sigma(X^*)+\rhoc(E\setminus X^*;p)\right)\\
&>0.
\end{align*}
The three equalities follow from~\eqref{mu_rho},
Lemma~\ref{lem:perturbation_decomposition},
and the modularity of $\sigma$, respectively.
The strict inequality holds by the minimality of $X^*$
and $X_2\subsetneq X^*$.
Thus, $X_1$ is over-demanded at $\bar p$, contradicting
Lemma~\ref{price_characterization}.
Hence, $p^*\leq\bar p$ for every equilibrium price vector
$\bar p$.
\end{proof}

\subsection{Proof of Lemma~\ref{any_direction}}
\label{subsec:any-direction-proof}

\begin{proof}
Fix sufficiently small $\varepsilon>0$.
Write every $Y\subseteq E$ as $Y_1\cup Y_2$, where
$Y_1\subseteq X^*$ and $Y_2\subseteq E\setminus X^*$.
By Lemma~\ref{lem:perturbation_decomposition}, we have
\begin{align*}
\sigma(Y)+\rhoc(E\setminus Y;p+\varepsilon d)
=\bigl[\sigma(Y_1)
  +\rhoc(X^*\setminus Y_1;p+\varepsilon d)\bigr]+\bigl[\sigma(Y_2)
  +\rhoc((E\setminus X^*)\setminus Y_2;p)\bigr].
\end{align*}
The second bracket equals the objective in~\eqref{argmin}
at $X^*\cup Y_2$ and price $p$, minus $\sigma(X^*)$.
Since $X^*$ minimizes that objective, this bracket is
minimized at $Y_2=\emptyset$, which is therefore
its inclusionwise minimal minimizer.
Hence, \LSC holds if and only if $X^*$ is the unique
minimizer of the first term, or equivalently,
$\sigma(X^*)<\sigma(Y_1)+\rhoc(X^*\setminus Y_1;p+\varepsilon d)$
for any $Y_1\subsetneq X^*$.
This is equivalent to
$\rhoc(X^*\setminus Y_1;p+\varepsilon d)>\sigma(X^*\setminus Y_1)$.
Since $X^*\setminus Y_1$ ranges over all nonempty subsets
of $X^*$, this is precisely condition~(b).
\end{proof}

\subsection{Proof of Theorem~\ref{convex-dual}}
\label{subsec:convex-dual-proof}
We first identify the componentwise minimal optimal
solution at $\gamma=0$.
We then establish a least optimal solution for the
lexicographic problem and identify its primary component.
\begin{lemma}
\label{gamma_0}
At $\gamma=0$, the vector $\chi_{\scriptscriptstyle X^*}$
is the unique componentwise minimal optimal solution of the problem $(D)$.
\end{lemma}

\begin{proof}
Define $f(X)\coloneqq\sigma(X)+\rhoc(E\setminus X;p)$ for $X\subseteq E$. At $\gamma=0$, Section~\ref{subsec:lexico-dual} gives $\dom\cL=[0,1]^E$. Greedy optimization over each $P(\rhoc_i)$ shows that $\cL$ is the Lov\'asz extension of $f$ on this box~\citep[Section~6.3]{F2005}. 
For any $z\in[0,1]^E$, take a chain decomposition $z=\sum_k\alpha_k\chi_{U_k}$ with $\alpha_k>0$ and $\sum_k\alpha_k=1$, allowing $U_k=\emptyset$. Then, $\cL(z)=\sum_k\alpha_k f(U_k)$. Thus, $z$ is optimal if and only if every $U_k$ minimizes $f$. Since $X^*$ is the unique inclusionwise minimal solution of~\eqref{argmin}, every such $U_k$ contains $X^*$. Hence, every optimal solution $z$ satisfies
$z\geq\chi_{\scriptscriptstyle X^*}$ componentwise. Since $\chi_{\scriptscriptstyle X^*}$ is itself optimal, it is the unique componentwise minimal optimal solution.
\end{proof}

\begin{proof}[Proof of Theorem~\ref{convex-dual}]
By Lemma~\ref{gamma_0} and Theorem~\ref{lexico_characterization},
$(D)$ has an optimal solution
$\widehat z=\chi_{X^*}+\gamma\widehat t$.
Let $\mathcal Z$ denote its optimal solution set.
By conjugacy~\citep[Theorem~8.4]{M2003} and the closure
properties of $L^\natural$-convexity,
$\mathcal Z$ forms a lattice under componentwise
minimum and maximum.
On $\dom\cL$, each supremum is a maximum over the finitely
many vertices of $B(\rhoc_i)$.
Thus, $\mathcal Z$ has a finite linear description
with real coefficients and lexicographic right-hand sides.
Since $\mathcal Z$ is nonempty and $z\geq\zeros$ on
$\mathcal Z$, each coordinate attains its minimum
by linear programming over lexicographic pairs.
The meet of these coordinate minimizers is the unique
least optimal solution $z^*$.

Write $z^*=z^0+\gamma t$.
By Theorem~\ref{lexico_characterization},
$z^0$ is optimal at $\gamma=0$.
Lemma~\ref{gamma_0} therefore gives $z^0\geq\chi_{X^*}$.
Since $z^*\leq\widehat z$, we also have
$z^0\leq\chi_{X^*}$.
Hence, $z^*=\chi_{X^*}+\gamma t$.
\end{proof}

\subsection{Proof of Lemma~\ref{lem:boundary-oracles}}
\label{app:boundary-oracles}

We construct value oracles for the boundary functions
used in direction computation and allocation recovery.
For each formulation, we provide an integral base point
and an exchange capacity oracle, which together enable
rank evaluation.
The following lemma includes Lemma~\ref{lem:boundary-oracles}
as a special case.
\begin{lemma}
    Let $\rho^+$ be the function defined in Section~\ref{subsec:lexico-dual} or~\ref{Allocation}.
    Given the demand and exchange capacity oracles for all agents, the value $\rho^+(S)$ for $S \subseteq \mathcal{A}$ can be computed in strongly polynomial~time.
\end{lemma}

\begin{proof}
Let $B\coloneqq B(\rho^+)\cap\Z^{\cA}$.
A point of $B$ and an exchange capacity oracle suffice to evaluate $\rho^+(S)=\max\Set{y(S)}{y\in B}$ in strongly polynomial time by linear optimization~\citep{OShi2026}.

In \Cref{subsec:lexico-dual}, $B=\prod_{i\in N}\Tc_i(p)$.
Starting from bundles returned by the demand oracles, maximally decrease each coordinate once using the exchange capacity oracle.
By the \Mn-convexity of $T_i(p)$, this yields a point of $B$.
Exchanges within each agent preserve the minimum component sum and have the same capacities as in $T_i(p)$.
Exchanges between agents have zero capacity.

In \Cref{Allocation}, $B=\mathcal{P}^+\cap\Z^{\cA}$, where $\mathcal{P}^+$ is~\eqref{eq:Pplus}.
Starting from bundles returned by the demand oracles, adjust $y(\cA)$ to $\sigma(E)$ one coordinate at a time, taking the largest feasible step without passing the target, as in~\citet[Appendix~B.3 in the full version]{OS2026}.
This takes at most $\abs{\cA}$ moves and succeeds since $\mu(E;p^*)\leq\sigma(E)\leq\rhoh(E;p^*)$ by Lemma~\ref{price_characterization}.
At $y\in B$, the exchange capacity from $(i,e)$ to $(j,f)$ is the oracle output for $(y_i,e,f)$ if $i=j$, and the minimum of the outputs for $(y_i,e,0)$ and $(y_j,0,f)$ otherwise.
\end{proof}

\subsection{Proof of Finite Termination in Theorem~\ref{thm:correctness_algo1}}
\label{Convergence}

We first show that every nonterminal update has
positive finite length.
We then bound the number of changes in $X^*$ and show
that only finitely many updates can occur while $X^*$
remains~fixed.

\begin{lemma}
\label{lem:positive-finite-step}
At every nonterminal iteration of Algorithm~\ref{algo1},
including its transformed implementation in Section~\ref{sec:sy},
the step size satisfies
$0<\lambda_{\mathrm{evt}}(p,d^*)<+\infty$.
\end{lemma}

\begin{proof}
We treat trading networks and the transformed market
with substitutes and complements using common notation.
For the latter, use the transfers and their linear
extensions in Section~\ref{sec:sy}.

Piecewise linearity of the transfers and finiteness of
the feasible sets imply that, for some $\delta>0$,
demands are constant for $0<t<\delta$ and transfer
slopes are unchanged for $0\leq t<\delta$
along $p+td^*$.
Hence, $\lambda_{\mathrm{evt}}(p,d^*)\geq\delta>0$.

In either model, choose $b_i\in\dom\tilde v_i$ with
$\sum_{i\in N}b_i=\sigma$:
take $b_i=\zeros$ in trading networks, or twist an allocation
assigning all supply to one bidder in the single-seller market.
Suppose to the contrary that
$\lambda_{\mathrm{evt}}(p,d^*)=+\infty$.
Set $X\coloneqq\supp_+(d^*)\neq\emptyset$.
Since demand sets are unchanged, \LSC gives
$X=X^*(p+td^*)$ for every $t>0$.
For each agent $i$, choose a bundle $y_i$ demanded
at every price $p+td^*$ with $t>0$. Then, 
\[
\sum_{i\in N}y_i(X)
\geq\mu(X;p+td^*)
>\sigma(X)
=\sum_{i\in N}b_i(X).
\]
Thus, $y_i(X)>b_i(X)$ for some agent $i$.

For this agent, $y_i$ does not minimize the sum of
coordinates in $X$ over the \Mn-convex set $\dom\tilde v_i$.
The local optimality criterion for \Mn-convex functions \citep[Theorem~4.4(2)]{M2022}
therefore yields a feasible unit move that decreases this sum.
Such a move must have the form
$z\coloneqq y_i-\chi_e+\chi_f\in\dom\tilde v_i$
with $e\in X$ and $f\in(E\setminus X)\cup\{0\}$.
The transfer slopes remain unchanged along $p+td^*$
for all $t\geq0$.
Since $d_f^*=0$ when $f\neq0$, we have
$\tilde u_i(z;p+td^*)-\tilde u_i(y_i;p+td^*) = \tilde u_i(z;p)-\tilde u_i(y_i;p)
+tq'_{ie}(p_e)d_e^*$. 
As $q'_{ie}(p_e)d_e^*>0$, this tends to $+\infty$,
contradicting the optimality of $y_i$ for all $t>0$.
Hence, $\lambda_{\mathrm{evt}}(p,d^*)<+\infty$.
\end{proof}

Let $z^*(p)$ be the unique minimal optimal solution of $(D)$
used to construct $d^*$, and let $V(p)$ be the common
optimal value of $(P)$ and $(D)$.
Demand sets may shrink immediately after an update begins.
The next lemma shows that $V$ and $z^*$ nevertheless
remain unchanged before the next iterate.
\begin{lemma}
\label{lem:update-stability}
Consider an update from $p$ to $p'=p+\lambda d^*$
in Algorithm~\ref{algo1}.
The optimal value and the minimal optimal dual solution
remain unchanged before reaching $p'$:
\[
V(p+sd^*)=V(p),
\qquad
z^*(p+sd^*)=z^*(p)
\quad(0\leq s<\lambda).
\]
\end{lemma}

\begin{proof}
The dual problem is unchanged for $0<s<\lambda$,
since demand sets and transfer slopes remain fixed.
It therefore suffices to compare the problems at $p$
and $p+sd^*$ for sufficiently small $s>0$.
Write $z^*(p)=\chi_{\scriptscriptstyle X^*}+\gamma t$.
Since the transfer slopes are unchanged, the two problems
have the same effective domain.
On this domain, each supremum in either dual objective
is attained on the corresponding base polyhedron.

Fix the primary component at $\chi_{\scriptscriptstyle X^*}$.
For each agent, consider the vertices maximizing the
primary objective.
By Lemma~\ref{greedy_structure}, all vertices attaining
the secondary maximum at $t$ remain in the updated
base polyhedron.
Every removed vertex has a strictly smaller secondary
objective value at $t$.
Since these values depend continuously on the secondary
component and there are finitely many vertices,
the removed vertices remain suboptimal in a neighborhood
of $t$.
Thus, removing them does not change the maximum
in this neighborhood.
The two dual objectives therefore agree there
within their common effective domain.

By \LSC and Theorem~\ref{convex-dual},
the updated minimal optimal solution $\widehat z$
has the same primary component as $z^*(p)$.
Local agreement and convexity imply that $z^*(p)$
minimizes the updated objective among points with
primary component $\chi_{\scriptscriptstyle X^*}$.
Since $\widehat z$ also has this primary component,
$z^*(p)$ is optimal for the updated problem,
with value $V(p)$. 
By minimality, $\widehat z\leq z^*(p)$. 
If $\widehat z\neq z^*(p)$, interior points of the line
segment joining them and sufficiently close to $z^*(p)$
are optimal for both problems, contradicting
the minimality of $z^*(p)$.
Hence, $\widehat z=z^*(p)$, proving both claims.
\end{proof}

Define the maximum over-demand value by
\[
\omega(p)
\coloneqq\mu(X^*(p);p)-\sigma(X^*(p))
=\max_{X\subseteq E}\bigl(\mu(X;p)-\sigma(X)\bigr).
\]
We first bound the decreases in the 
over-demand value.
\begin{lemma}
\label{lem:over-demand-progress}
Throughout Algorithm~\ref{algo1}, $\omega(p)$ is nonincreasing,
with at most $\omega(\zeros)$ strict decreases.
\end{lemma}

\begin{proof}
Consider an update from $p$ to $p'=p+\lambda d^*$.
Since demand sets remain unchanged for $0<s<\lambda$,
\LSC gives $X^*(p+sd^*)=X^*(p)$ throughout this interval.
Since prices outside $X^*(p)$ remain fixed,
Lemma~\ref{mu_rho_relation} with $Y=E\setminus X^*(p)$
shows that $\omega(p+sd^*)$ is nonincreasing for
$0\leq s<\lambda$.

Choose a sufficiently small $\eta\in(0,\lambda)$
and set $p^-\coloneqq p'-\eta d^*$.
Finiteness of the feasible sets and continuity of utilities
give $T_i(p^-)\subseteq T_i(p')$ for every $i$.
Hence, $\mu(X;p')\leq\mu(X;p^-)$ for all $X\subseteq E$.
Subtracting $\sigma(X)$ and maximizing over $X$ gives
$\omega(p')\leq\omega(p^-)\leq\omega(p)$.
Thus, $\omega$ is nonincreasing.
Since it is nonnegative and integer-valued,
it strictly decreases at most $\omega(\zeros)$ times.
\end{proof}

We next bound the number of changes in $X^*$.
\begin{lemma}
\label{event_a}
During Algorithm~\ref{algo1}, the set $X^*(p)$ changes
at most $O\bigl((n\,c(E)+\abs{\sigma}(E))m\bigr)$ times,
where $\abs{\sigma}(E)\coloneqq\sum_{e\in E}\abs{\sigma_e}$.
\end{lemma}

\begin{proof}
Suppose $X^*$ changes at price $p$ without a decrease in $\omega$.
Let $p^-\coloneqq p-\eta d^*$ for sufficiently small $\eta>0$,
where $d^*$ is the direction of the update ending at $p$.
As above, $T_i(p^-)\subseteq T_i(p)$ for every $i$.
Hence,
$\omega(p^-)
\geq \mu(X^*(p);p^-)-\sigma(X^*(p))
\geq \omega(p)
=\omega(p^-).$
Thus, $X^*(p)$ also maximizes
$\mu(\cdot;p^-)-\sigma(\cdot)$.
Since these maximizers form a lattice with minimal element
$X^*(p^-)$, and $X^*$ changes at $p$, we obtain 
$X^*(p^-)\subsetneq X^*(p)$.
For each fixed positive value of $\omega$, every change
strictly enlarges $X^*$.
There are therefore at most $m$ such changes.
By Lemma~\ref{lem:over-demand-progress},
$\omega$ takes at most $\omega(\zeros)$ positive values
and strictly decreases at most $\omega(\zeros)$ times.
Moreover, $\omega(p)=0$ implies $X^*(p)=\emptyset$,
at which point the algorithm terminates.
Hence, $X^*$ changes at most $(m+1)\omega(\zeros)$ times.
Since $y_{ie}\leq c_e$ for every feasible bundle,
we have $\mu(X;\zeros)\leq n\,c(E)$ for all $X\subseteq E$.
Using $-\sigma(X)\leq\abs{\sigma}(E)$ and maximizing over $X$
gives $\omega(\zeros)\leq n\,c(E)+\abs{\sigma}(E)$,
which proves the claimed~bound.
\end{proof}

It remains to bound the number of updates while $X^*$ is fixed.

\begin{lemma}
\label{event_b}
On any interval where $X^*(p)$ remains fixed,
Algorithm~\ref{algo1} performs only finitely many updates.
\end{lemma}

\begin{proof}
Since prices are nondecreasing and each transfer function
has finitely many linear pieces, transfer slopes change
only finitely many times.
It therefore suffices to consider an interval on which
$X^*$ and all transfer slopes are fixed. 
The weights $w$ and the effective domain
$\mathcal Z\coloneqq\operatorname{dom}\mathcal L$ are then fixed.
Since $\omega$ changes
only finitely many times by Lemma~\ref{lem:over-demand-progress}, 
we may further assume that it is constant.
At a demand event $p$, set $p^-\coloneqq p-\eta d^*$
for sufficiently small $\eta>0$, where $d^*$ is the direction
of the update ending at $p$.
As above, $T_i(p^-)\subseteq T_i(p)$ for every $i$.
These inclusions give $\mu_i(X^*;p^-)\geq\mu_i(X^*;p)$
for every $i$.
Since $\omega$ and $X^*$ are fixed, the sums of these
quantities agree, forcing equality for every $i$. 
Lemma~\ref{mu_rho_relation} with $Y=E\setminus X^*$
then gives $\rhoc_i(E;p^-)=\rhoc_i(E;p)$. 
Thus, the minimum total sums are equal, yielding
$\Tc_i(p^-)\subseteq\Tc_i(p)$ and 
$P\bigl(\rhoc_i(\cdot;p^-)\bigr)
\subseteq P\bigl(\rhoc_i(\cdot;p)\bigr)$.

These inclusions expand the feasible region of $(P)$
without changing its objective, giving $V(p^-)\leq V(p)$.
Together with Lemma~\ref{lem:update-stability},
this shows that $V$ is nondecreasing.
Since each demand set lies in a finite twisted box
and the slopes are fixed, only finitely many instances
of $(P)$ and $(D)$ are possible.
Thus, $V$ strictly increases only finitely many times.

We may therefore further restrict to a subinterval
on which $V$ is constant.
Recall that $z^*(p)$ denotes the componentwise minimal
optimal solution of $(D)$ at price $p$.
At a demand event $p$, the preceding inclusions give
$\mathcal L_{p^-}(z)\leq\mathcal L_p(z)$ for all $z\in\mathcal Z$,
where $\mathcal L_{p^-}$ and $\mathcal L_p$ are the dual
objectives at $p^-$ and $p$, respectively.
Hence, 
$V(p)
\leq\mathcal L_{p^-}\bigl(z^*(p)\bigr)
\leq\mathcal L_p\bigl(z^*(p)\bigr)
=V(p)$.
Thus, $z^*(p)$ is also optimal at $p^-$, implying
$z^*(p^-)\leq z^*(p)$.
Lemma~\ref{lem:update-stability} extends this monotonicity
to the entire subinterval.
Since $z^*$ is nondecreasing and takes only finitely
many values, it changes only finitely many times.

Finally, while $z^*$ is constant, $d^*$ is fixed.
With fixed transfer slopes, the finitely many feasible
bundles have affine utilities along the price path,
yielding finitely many demand events.
By Lemma~\ref{lem:positive-finite-step}, each nonterminal
update reaches the next event.
Hence, each subinterval contains finitely many updates.
Since there are finitely many such subintervals,
the claim follows.
\end{proof}
Lemmas~\ref{event_a} and \ref{event_b}
together imply that Algorithm~\ref{algo1}
terminates after finitely many iterations.

\end{document}